%% file: Option_Pricing_Compound_Carma_Hawkes_v1.0.tex
\documentclass[AMS,Times1COL]{WileyNJDv6} 
\usepackage{amssymb,amsmath,amsthm,array,bm}

\usepackage{epstopdf}
\usepackage{graphicx}
\usepackage{color,bbm}
\usepackage{subcaption}
\DeclareMathOperator*{\argmin}{arg\,min}

\newcommand{\ft}{f_{t}\left(\ln\left(S_{t}\right),X_{t},m_{t}\right)}
\newcommand{\ftj}{f_{t}\left(\ln\left(S_{t}\right),X_{t}+\mathbf{e},m_{t}\right)}
\newcommand{\ftzero}{f_{t_0}\left(\ln\left(S_{t_0}\right),X_{t_0},m_{t_0}\right)}
\newcommand{\ftm}{f_{t}\left(\ln\left(S_{t^{-}}\right),X_{t},m_{t^{-}}\right)}

\newcommand{\cflp}{\mathbb{E}^{\mathbb{Q}}\left[e^{iu\ln\left(S_{T}\right)}\left|\mathcal{F}_{t_0}\right.\right]}

\articletype{Option Pricing with a Compound CARMA(p,q)-Hawkes}%

\received{Date Month Year}
\revised{Date Month Year}
\accepted{Date Month Year}
\journal{Journal}
\volume{00}
\copyyear{2023}
\startpage{1}

\begin{document}

\title{Option Pricing with a Compound CARMA(p,q)-Hawkes}

\author[1]{Lorenzo Mercuri}

\author[1]{Andrea Perchiazzo}

\author[2]{Edit Rroji}

\authormark{Mercuri \textsc{et al.}}
\titlemark{Option Pricing with a Compound CARMA(p,q)-Hawkes}

\address[1]{\orgdiv{Department of Economics, Management and Quantitative Methods}, \orgname{University of Milan}, \orgaddress{\country{Italy}}}

\address[2]{\orgdiv{Department of Statistics and Quantitative Methods}, \orgname{University of Milano-Bicocca}, \orgaddress{\country{Italy}}}

\corres{Corresponding author Lorenzo Mercuri, This is sample corresponding address. \email{lorenzo.mercuri@unimi.it}}



\abstract[Abstract]{

A self-exciting point process with a continuous-time autoregressive moving average intensity process, named CARMA(p,q)-Hawkes model, has recently  been introduced. The model generalizes the Hawkes process by substituting the Ornstein-Uhlenbeck intensity with a CARMA(p,q) model where the associated state process is driven by the counting process itself. The proposed model preserves the same degree of tractability as the Hawkes process, but it can reproduce more complex time-dependent structures observed in several market data. The paper presents a new model of asset price dynamics based on the CARMA(p,q) Hawkes model. It is constructed using a compound version of it with a random jump size that is independent of both the counting and the intensity processes and can be employed as the main block for pure jump and (stochastic volatility) jump-diffusion processes. The numerical results for pricing European options illustrate that the new model can replicate the volatility smile observed in financial markets. Through an empirical analysis, which is presented as a calibration exercise, we highlight the role of higher order autoregressive and moving average parameters in pricing options.
}



\maketitle


\renewcommand\thefootnote{\arabic{footnote}}
\setcounter{footnote}{-0}

\section{Introduction}\label{intro}
\input{introduction}

\section{Asset price dynamics with a compound CARMA(p,q)-Hawkes jump component}\label{sec2}

\input{stochastic_model}

\section{Gauss-Laguerre quadrature for option pricing}
\label{GaussLaguerre}

\input{GLmethod}


\section{Simulation algorithm}\label{sim_algorithm}

\input{simulation_file}

\section{Numerical results}\label{num_analysis}

\input{numerical_results}
\section{Empirical analysis}\label{emp_analysis}

\input{empirical_analysis}

\newpage
\section{Conclusions}\label{conclusions} 
This paper introduces a novel approach for the asset price dynamics using a compound CARMA(p,q)-Hawkes process that provides a flexible framework for capturing jump dynamics under both real-world and risk-neutral probability measures. We also provide a change of measure that distorts only the jump size distribution while the log-affine structure of the log-price process is preserved under the risk-neutral measure.
Additionally, the paper introduces a numerical method for pricing European options using the Gauss-Laguerre quadrature, which offers a good level of efficiency and stability. The accuracy of the proposed approach is supported by numerical results, where option prices are compared with those obtained through Monte Carlo simulations.
Finally, a simulation algorithm is developed  for the CARMA(p,q)-Hawkes jump-diffusion. This algorithm turns out to be computationally efficient and it can be used for pricing path-dependent derivatives.


\bmsection*{Author contributions}
All authors wrote the main manuscript text and reviewed the manuscript

\bmsection*{Acknowledgments}
This work was supported by JST CREST Grant Number JPMJCR2115, Japan and by the PRIN2022 project ``The effects of climate change in the evaluation of financial instruments'' financed by the ``Ministero dell’Università e della Ricerca'' with grant number 20225PC98R, CUP Codes: H53D23002200006 (University of Milano-Bicocca) and G53D25001960006 (University of Milan).
Lorenzo Mercuri and Andrea Perchiazzo acknowledge financial support within the ``Fund for Departments of Excellence academic funding'' provided by the Ministero dell'Università e della Ricerca (MUR), established by Stability Law, namely `Legge di Stabilità n.232/2016, 2017' - Project of the Department of Economics, Management, and Quantitative Methods, University of Milan.

%

\bmsection*{Conflict of interest}
The authors declare no potential conflict of interests.
%

\bmsection*{Supporting information}

Additional supporting information may be found in the
online version of the article at the publisher’s website.

\appendix

\section{CARMA(p,q)-Hawkes review and Toy Model}\label{app1}
\vspace*{12pt}

\input{Carma_review}

\end{document}

%% file: introduction.tex
In both academic and financial industry literature it is well established that the Black-Scholes model is inconsistent with empirical observations; e.g., the volatility smile, skewness, and sudden large price fluctuations (intended as jumps in prices as done in \cite{arias2024heston}). Several models have been developed and studied in order to address these limitations and to ensure consistency with market data. Examples of these stochastic processes, but are not limited to, include local volatility models, stochastic volatility models, and (stochastic volatility) jump-diffusion models. In this paper we focus on the framework of financial models with jumps. 

It is widely acknowledged that the role of jumps is of great importance within the context of financial market dynamics. To illustrate, as explained in \cite[Section 1.1]{tankov2015financial}, in the event of sudden price movements, a diffusion model is unable to capture this behaviour (i.e., the presence of jumps in prices), whereas a model with jumps can. Moreover, the pricing of derivatives is significantly influenced by the presence of jumps. Models incorporating jumps are able to replicate the implied volatility smile and skew that is observed in financial markets. In this context, two notable milestones are the jump-diffusion processes proposed by Merton \cite{merton1976option} and Kou \cite{kou2002jump, kou2004option}. However, as pointed out in \cite{du2019pricing} and \cite{jing2021consistent}, standard jump-diffusion models\footnote{For the sake of clarity, the term ``standard jump-diffusions'' is used to describe a model composed of two components: a diffusive part and a jump part. Here the jumps represent rare events that can be modelled with a compound Poisson process with constant intensity. This implies that increments are independent and identically distributed (i.i.d.),  precluding the possibility of serial dependence in jumps.
A list of standard jump-diffusion models is presented in \cite{du2019pricing}; here, we report \cite{bates1996jumps, bates2000post, pan2002jump, broadie2007model}.} are insufficient for accurately representing empirical observations due to their inability to account for the clustering behaviour of jumps; that is, the phenomenon in which large and sudden  fluctuations in prices, frequently triggered by unexpected events, tend to concentrate (cluster) in short time intervals. As explained in \cite{fulop2015self, du2019pricing, jing2021consistent},  the occurrence of a single jump appears to stimulate the occurrence of subsequent jumps (i.e., the probability of future jumps increases, thereby giving rise to a self-excitation effect), culminating in the formation of clusters of jumps (especially pronounced during periods of financial turbulence). In the presence of jump propagation/contagion, future jumps would affect option trading in a manner distinct from the future jumps under traditional jump-diffusion models. For example, as argued in \cite{du2019pricing}, the jump propagation/contagion effect prompts forward-looking investors to perceive future jumps as a trigger for a severe financial turmoil (or market collapse). Consequently, as the expectation of occurrence of jumps increases, investors become motivated to purchase deep out-of-the-money options as a hedging strategy. 

Due to the fact that standard jump-diffusion models are not well-equipped to adequately represent the self-excitation effect and the clustering pattern of jumps observed in market data, many authors have chosen to employ the Hawkes process \cite{Hawkes1971}, a parsimonious and self-exciting model, as an alternative. In the literature, the Hawkes process has been used and extended for the purpose of a) modelling the default risk of a credit portfolio \cite{errais2010affine}; b) market microstructure analysis and high-frequency data \cite{bacry2013modelling, bacry2014hawkes, maneesoonthorn2017inference}; c) modelling interest rates \cite{hainaut2016model, hainaut2016bivariate}; d) simultaneous modelling of market indices \cite{ait2015modeling}; e) and derivatives pricing and hedging \cite{ma2017pricing, hainaut2018hedging, brignone2020asian, arias2024heston}.  For a survey regarding the Hawkes process, the Hawkes jump-diffusions model, and its application in finance, refer to \cite{hawkes2018hawkes, hawkes2022hawkes}.

Although the Hawkes process with an exponential kernel is Markovian and exhibits a high degree of tractability that thereby makes it a valuable tool for real applications (see \cite{errais2010affine}), the specific form of the kernel constraints the shape of the temporal dependence structure of the number of jumps observed in intervals of equal length. Specifically, the autocorrelation in a Hawkes model is a decaying function of the time lags (refer to \cite{da2014hawkes}). This makes the exponential Hawkes model unsuitable for capturing the dependence structures observed in a multitude of real-world data sets (e.g., wind speed data \cite{benth2019non}, mortality rates \cite{hitaj2019levy}, earthquake and intra-day orders of an Italian government bond data \cite{MERCURI20241}). It should be noted that the empirical analysis presented in this work similarly reveals that the Hawkes jump-diffusion model is not able to accurately represent empirical data, especially deep out-of-the-money options in which there is a high level of trading volume. 

In consideration of the shortcomings of the Hawkes model, an extension and generalisation of the aforementioned self-exciting point process, namely the CARMA(p,q)-Hawkes model \cite{MERCURI20241}, has been proposed in the literature as a potential alternative. The CARMA(p,q)-Hawkes model is a Hawkes process in which the intensity of the model is governed by a continuous-time autoregressive moving average (CARMA) process,\footnote{For a survey regarding CARMA(p,q) models and its applications, we suggest \cite{doob1944elementary, brockwell2001levy, brockwell2004representations, andresen2014carma, tomasson2015some, mercuri2021finite}.} rather than by an Ornstein–Uhlenbeck model. In order to generate the self-exciting effect, the driving noise of the CARMA(p,q) model is the counting process itself. The model maintains a similar analytical structure of the Hawkes process with the exponential kernel that allows to compute, for example, the infinitesimal generator, the high-order moments, the autocorrelation function, and the joint characteristic function. However, as documented in \cite{MERCURI20241}, it is capable of reproducing a greater range of time-dependence structures observed in the data due to the fact that the kernel is no longer exponentially decreasing. Starting from the CARMA(p,q)-Hawkes model, this paper introduces a new model for asset price dynamics within the context of financial models with jumps and is constructed using a compound CARMA(p,q)-Hawkes model with a random jump size independent of both the counting process and the intensity process. 
The novel model can be used as the basis for the construction of a pure jump model, a jump-diffusion model (e.g., a CARMA(p,q)-Hawkes jump-diffusion model in contrast to well-known Hawkes jump-diffusion model), and a stochastic volatility jump-diffusion model (e.g., Heston CARMA(p,q)-Hawkes jump-diffusion model), serving as the underlying asset price dynamics in a new option pricing model. 

The contribution of this paper is threefold. Firstly, it presents a novel financial model based on CARMA(p,q)-Hawkes jumps and a framework that encompasses both the real-world probability and risk-neutral probability measures. Specifically, we propose a dynamics for the log-price that has log-affine characteristic function under the real-world probability measure. Given that the market in question is incomplete, we introduce a change of measure that produces a distortion in the distribution of the jump size that is exactly the Esscher transform \cite{gerber1993option} of the jump size. The introduced change of measure preserves the structure of the log-price under the risk-neutral probability measure in which the characteristic function of the log-price remains log-affine. 
We expect that, for a given distribution of the jump size, the model with a compound CARMA(p,q)-Hawkes is able to capture the features in option prices due to the more complex memory structure of time arrivals of jumps. 

Secondly, a new approach for pricing European options using the characteristic function is introduced. Instead of the Carr-Madan formula \cite{carr1999option}, which requires a suitable choice of the damping factor, or the COS method \cite{fang2009novel} that introduces three approximation errors (see \cite[Chapter 6.2.3]{oosterlee2019mathematical} for the error analysis), the proposed approach for computing European option prices is based on the Gauss-Laguerre quadrature. The approximation error is controlled by the order of Laguerre polynomials and the approach appears to be stable and less time-consuming in the calibration exercise.

In order to assess the effectiveness of the proposed approach, option prices obtained through the Gauss-Laguerre quadrature are compared with Monte Carlo prices. Thirdly, we present a simulation algorithm for our model, which is inspired by the thinning algorithm proposed in \cite{Ogata1981}. This algorithm is based on the possibility of bounding from above the intensity of a CARMA(p,q)-Hawkes by the intensity of a Hawkes with exponential kernel. 

For the sake of clarity, we would like to reiterate that the present work restricts its focus to the structure of the time arrivals of jumps. As a matter of fact, we wish to emphasise that all ingredients introduced in the paper, namely the compound CARMA(p,q)-Hawkes process, the Radon-Nikodym derivative, the Gauss-Laguerre method for pricing European options numerically, and the simulation algorithm, can be readily applied to the case of (stochastic volatility) jump-diffusion processes where the jumps are modelled by a compound CARMA(p,q)-Hawkes process. 

In light of the aforementioned considerations, the paper is organised as follows. Section~\ref{sec2} presents the framework for the construction of a financial model with jumps using a compound CARMA(p,q)-Hawkes model. It specifically shows the dynamics for the log-price process and its characteristic function under the real-world probability measure, the change of measure, and finally, the dynamics for the log-price process and its characteristic function under the risk-neutral probability measure. The diffusion component is excluded from this section as the primary emphasis is on the jump component. This is due to the fact that the construction of a jump-diffusion model is a relatively straightforward calculation. 
Finally, we generalize the pricing model by adding a diffusion term. 
In Section~\ref{GaussLaguerre}, the novel approach for pricing European options using the characteristic function through the Gauss-Laguerre quadrature is introduced. Section~\ref{sim_algorithm} discusses a numerical method for simulating the underlying asset, modelled as a CARMA(p,q)-Hawkes jump-diffusion process, under the risk-neutral probability measure in which an efficient thinning simulation algorithm for the CARMA(p,q)-Hawkes model is developed. In Section~\ref{num_analysis}, we present the numerical results in terms of call option prices obtained through the Gauss-Laguerre approach and the Monte Carlo simulation, as well as the implied volatilities by means of the Gauss-Laguerre approach, for Hawkes jump-diffusion, CARMA(2,1)-Hawkes jump-diffusion, and CARMA(3,1)-Hawkes jump-diffusion models that confirm the theoretical results. Furthermore, a sensitivity analysis is performed. Section~\ref{emp_analysis} presents the empirical results using options written on \textit{GameStop} company and Section~\ref{conclusions} concludes. Appendix~\ref{app1} contains a simplified version of the model presented in Section~\ref{sec2} where the expectation of the exponential of jump size is constrained. In this simplified version, a closed-form solution for option prices similar to Merton's jump-diffusion model is available. Three distributions for the jump size are considered.

%% file: stochastic_model.tex
Let $\left(\Omega,\mathcal{F},\mathbb{F},\mathbb{P}\right)$ be a filtered probability space with a filtration $\mathbb{F}=\left(\mathcal{F}_{t}\right)_{t\in \left[0,\mathbb{T}\right]}$ with $\mathbb{T}>0$ where all the processes defined below are adapted. Consider an underlying asset with price $S_t$ that, under the real-world probability measure $\mathbb{P}$, satisfies the following stochastic differential equation (SDE):
\begin{equation}
	\mbox{d}S_t=\alpha_t S_{t^{-}} \mbox{d}t+S_{t^{-}} \mbox{d}\tilde{Y}_t
	\label{pureJumpwithnounitary}
\end{equation}
given an initial condition for $S_{t_0}$ where $t_0 \in \left[0,\mathbb{T}\right)$ and $\left\{\alpha_{t}\right\}_{t \in \left[0,\mathbb{T}\right]}$ is a predictable process.\footnote{Note that while we assume that the price originates from a dynamic that starts at time 0, $t_0$ will coincide with the day of the market quote.}
The c\`adl\`ag stochastic process $\left\{\tilde{Y}_t\right\}_{t \in \left[0,\mathbb{T}\right]}$, named compound CARMA(p,q)-Hawkes, has 
piecewise sample paths starting at 0
and is defined as
\begin{equation}
\tilde{Y}_t =\sum_{k=1}^{N_t}\epsilon_k, 
\label{CompCH}
\end{equation}
where $N_{t}$ is a counting process and $\left\{\epsilon_{k}\right\}_{k\geq 1}$ is a sequence of i.i.d. random variables that represents the jump size in the price dynamics. The processes $N_{t}$ and $\left\{\epsilon_{k}\right\}_{k\geq 1}$ are independent and $\epsilon_k \overset{d}{=} \bar{\epsilon}, \forall k \in \mathbb{N}_{+} $. The domain of the distribution $\bar{\epsilon}$ is bounded from below, i.e., $\left(-1,+\infty\right)$. \newline
Let $\left\{J_{k}\right\}_{k\geq 1}$ be a sequence representing the jump size in the log-price. The quantity $J_k$ can be defined using $\epsilon_{k}$ as follows:\footnote{The sequence of $\left\{J_{k}\right\}_{k\geq 1}$ inherits the independence structure of $\left\{\epsilon_{k}\right\}_{k\geq 1}$ and  the distribution of each $J_k$ is obtained from the distribution of $\bar{\epsilon}$ as
$J_k \overset{d}{=} \bar{J}:=\ln\left(1+\bar{\epsilon}\right)$.
}
\begin{equation}
J_k := \log\left(1+\epsilon_k\right).
\label{eq:J_k}
\end{equation}
The jump size of the log-price in \eqref{eq:J_k} is useful for the construction of the compound  CARMA(p,q)-Hawkes process $\left\{Y_{t}\right\}_{t\in\left[0,\mathbb{T}\right]}$ that drives the log-price dynamics, that is:
\begin{equation}
Y_t:=\sum_{k=1}^{N_t}J_k \text{ with } J_0=0.
\label{YtDef}
\end{equation}
Denoting by $\left\{T_i\right\}_{i\geq1}$ a non-decreasing non-negative process that represents the time arrival of a specific event, $N_t$ is defined as:
\begin{equation}
N_{t}:= \sum_{i\geq1}\mathbbm{1}_{\left\{T_i\leq t\right\}}
\label{Def:Nt}
\end{equation}
with conditional intensity process $\left\{\lambda_t\right\}_{t\in\left[0,\mathbb{T}\right]}$ modelled as:
\begin{equation}
\lambda_t = \mu+ \mathbf{b}^{\top}X_t,
\label{Def:Lambda}
\end{equation}
where $\mu\in \mathbb{R_+}$; the $p$-dimensional vector $\mathbf{b}$ contains the moving average parameters $\left\{b_i\right\}_{i=0}^{q}$ and has the following form:
\begin{equation}\label{bbold}
\mathbf{b}=\left\{
	\begin{array}{ll}
	\left[b_0,\ldots,b_{q}\right]^{\top} & \text{ if } p-q=1 \\
	\left[b_0,\ldots,b_{q},b_{q+1}, \ldots, b_{p-1}\right]^{\top}\text{, with } b_{q+1} = \ldots = b_{p-1}=0 &\text{ if } p-q\geq 2.
	\end{array}
	\right. 
\end{equation} 
The predictable state process $X_t$ has the following representation:\footnote{The process $\left\{X_{t}\right\}_{t\in\left[0,\mathbb{T}\right)}$ has c\`agl\`ad trajectories that implies $X_{t}=X_{t^{-}}:=\underset{h\rightarrow 0^-}{\lim}X_{t+h}$.}
\begin{equation}\label{eq:CarLambda}
X_t=\int_{\left[0,t\right)} e^{\mathbf{A}\left(t-s\right)}\mathbf{e} \mbox{d}N_s,
\end{equation}
where \(e^\mathbf{A}:=\underset{h=0}{\stackrel{+\infty}{\sum}}\frac{1}{h!}\mathbf{A}^h\). The process $X_t$ in \eqref{eq:CarLambda} satisfies the SDE:
\begin{equation}\label{eq:state_process}
\mbox{d}X_t=\mathbf{A}X_{t}\mbox{d}t+\mathbf{e}\mbox{d}N_t, \text{ with } X_{0}=\mathbf{0}.
\end{equation}  
The \(p\times p\) matrix $\mathbf{A}$, named companion matrix, has the following form: 
\begin{equation}
\mathbf{A}=\left[
\begin{array}{ccccc}
0 & 1 & 0 & \ldots & 0\\
0 & 0 & 1 & \ldots & 0\\
\vdots & \vdots & \vdots & \ddots & \vdots\\
0 & 0 & 0 & \ldots & 1\\
-a_p & -a_{p-1} & -a_{p-2} & \ldots & -a_1\\
\end{array}
\right]_{p\times p},
\label{Acomp}
\end{equation}
where $a_{1},\ldots a_{p}$ are the autoregressive parameters.
The $p$-dimensional column vector $\mathbf{e}$ is defined as: 
\begin{equation}
\mathbf{e}=\left[0,\ldots,1\right]^\top.
\label{ebold} 
\end{equation} 

The vector process $\left(X_t, N_T\right)$, introduced in \cite{MERCURI20241}, is named CARMA(p,q)-Hawkes model.

\begin{remark}\label{remark1}
In this framework we require that the conditions for the stationarity of the CARMA(p,q)-Hawkes and the positivity of the conditional intensity $\lambda_t$ to be satisfied; see \cite{MERCURI20241} for further details.
\end{remark}

Let $S_{t_0}$ be the price of the underlying asset at the day $t_0$ of the market quote and $T\in\left(t_0,\mathbb{T}\right]$ the maturity of a financial derivative. Using the Dol\'eans-Dade exponential \cite{doleans1970quelques}, the solution of \eqref{pureJumpwithnounitary} reads:
\begin{equation}
S_{T}=S_{t_0}\exp\left[\int_{t_0}^{T} \alpha_t \mbox{d}t + \int_{t_0}^{T} J_t \mbox{d}N_t\right].
\label{eq:sol}
\end{equation}

\begin{remark}\label{Remark2} 
The quantity \eqref{eq:sol} can be rewritten using the compound CARMA(p,q)-Hawkes increment over the interval $\left[t_0, T\right]$, thus we have:
\begin{equation*}
S_T=S_{t_0}\exp\left[\int_{t_0}^{T} \alpha_t \mbox{d}t + \left(Y_{T}-Y_{t_0}\right)\right],
\end{equation*}
where $Y_{t}$ is defined in \eqref{YtDef}.
\end{remark}
\begin{theorem}\label{characteristicP}
Let $r$ be the constant risk-free interest rate and $\varphi \in \mathbb{R}$. Assuming that $\alpha_t$ in \eqref{pureJumpwithnounitary} is defined as
\begin{equation}
\alpha_t:=r+\varphi\lambda_{t}, \quad \forall t\in\left[t_0, \ T\right],
\label{condAlphat}
\end{equation}
the conditional characteristic function of the log-price $S_T$ at maturity $T$ given $\mathcal{F}_{t_0}$ is log-affine in the couple $\left[\ln\left(S_{t_0}\right),X_{t_0}\right]$. That is, 
\begin{equation}\label{cf_theo}
\mathbb{E}\left[e^{iu_1\ln\left(S_{T}\right)}\left|\mathcal{F}_{t_0}\right.\right]:=\exp\left[u_{0,T}\left(t_0\right)+u_{1,T}\left(t_0\right)\ln\left(S_{t_0}\right)+u_{2,T}\left(t_0\right)^{\top}X_{t_0}\right],
\end{equation}
where $\forall t \in \left[t_0, \ T\right],$ $u_{1,T}\left(t\right) = i u_1$ and the time-dependent coefficients $u_{0,T}\left(t\right):\left[t_0, \ T\right]\rightarrow \mathbb{C}$ and $u_{2,T}\left(t\right):\left[t_0, \ T\right]\rightarrow \mathbb{C}^p$ satisfy the following system of ordinary differential equations:\footnote{Let $f\left(x\right): A \subseteq  \mathbb{R}^{n} \rightarrow B \subseteq  \mathbb{R}^{m}$ be a differentiable function.
We denote its Jacobian matrix by $\mathbb{J}_{f\left(x\right)}$.}
\begin{equation}
\left\{
\begin{array}{l}
\frac{\partial u_{0,T}\left(t\right)}{\partial t} = \mathbf{\mu}\left[1 - \mathbb{E}\left(e^{i u_1 \bar{J}}\right)e^{u_{2,T}\left(t\right)^{\top}\mathbf{e}}\right] - i u_1 \left(r+\varphi\mu\right) \\
\mathbb{J}_{u_{2,T}\left(t\right)}^{\top} = \left[1 - \mathbb{E}\left(e^{i u_1 \bar{J}}\right)e^{u_{2,T}\left(t\right)^{\top}\mathbf{e}}\right]\mathbf{b}^{\top} - u_{2,T}\left(t\right)^{\top}\mathbf{A} -i u_1 \varphi\mathbf{b}^{\top}
\end{array}
\right.,
\label{b01}
\end{equation}
with final conditions
\begin{equation}
\left\{
\begin{array}{l}
u_{0,T}\left(T\right)=0\\
u_{2,T}\left(T\right)= \mathbf{0}\\
\end{array}
\right. .
\label{finalCondLogPrice}
\end{equation}
\end{theorem} 
\begin{proof}
We consider the joint conditional characteristic function of $\left\{\ln\left(S_{t}\right),X_t\right\}_{t \in \left[t_0, \ T\right]}$ defined as:
\begin{equation}
f_{t}\left(\ln\left(S_T\right),X_T\right) = \mathbb{E}\left[\exp\left[i u_1\ln\left(S_T\right)+i u_2^{\top} X_{T}\right]\left|\mathcal{F}_{t}\right.\right].
\label{eq:jointCharLnST}
\end{equation}
The process $\left\{f_t\left(\ln\left(S_T\right),X_T\right)\right\}_{t \in \left[t_0, \ T\right]}$ is a complex-valued martingale and its dynamics can be determined using the It\^o's Lemma for semimartingales \cite{tankov2015financial} as follows:\footnote{By $\frac{\partial{f\left(x_{t^{-}},y_{t^{-}}\right)}}{\partial x_t}$ we denote the partial derivative of function $f\left(x_t,y_t\right)$ with respect to $x_t$ evaluated at $\left(x_{t^{-}},y_{t^{-}}\right)$, i.e.:
\begin{equation*}
\frac{\partial f\left(x_{t^{-}},y_{t^{-}}\right)}{\partial x_t}:=\left.\frac{\partial f\left(x_{t},y_{t}\right)}{\partial x_t}\right|_{\left(x_t,y_t\right)=\left(x_{t^{-}},y_{t^{-}}\right)}.
\end{equation*}
}
\begin{eqnarray}
\mbox{d}f_t\left(\ln\left(S_t\right),X_t\right) &=& \frac{\partial f_t\left(\ln\left(S_{t^{-}}\right),X_t\right)}{\partial t}\mbox{d}t + \frac{\partial f_t\left(\ln\left(S_{t^{-}}\right),X_t\right)}{\partial \ln\left(S_t\right)}\mbox{d}\ln\left(S_t\right)+\left[\nabla_{X_t}f_t\left(\ln\left(S_{t^{-}}\right),X_t\right)\right]^{\top}\mbox{d}X_{t}\nonumber\\
&+& \left[f_{t}\left(\ln\left(S_{t^{-}}\right)+J_t,X_t+\mathbf{e}\right)-f_{t}\left(\ln\left(S_{t^{-}}\right),X_t\right)-\frac{\partial f_t\left(\ln\left(S_{t^{-}}\right),X_t\right)}{\partial \ln\left(S_t\right)}J_t-\left[\nabla_{X_t}f_t\left(\ln\left(S_{t^{-}}\right),X_t\right)\right]^{\top}\mathbf{e}\right]\mbox{d}N_t\nonumber\\
&=&\left[\frac{\partial f_t\left(\ln\left(S_{t^{-}}\right),X_t\right)}{\partial t} + \frac{\partial f_t\left(\ln\left(S_{t^{-}}\right),X_t\right)}{\partial \ln\left(S_t\right)}\alpha_t + \left[\nabla_{X_t}f_t\left(\ln\left(S_{t^{-}}\right),X_t\right)\right]^{\top}\mathbf{A}X_{t}\right]\mbox{d}t\nonumber\\
&+& \left[f_{t}\left(\ln\left(S_{t^{-}}\right)+J_t,X_t+\mathbf{e}\right)-f_{t}\left(\ln\left(S_{t^{-}}\right),X_t\right)\right]\mbox{d}N_t.
\label{term_add_sub}
\end{eqnarray}
Adding and subtracting  the quantity $\left(\mathbf{\mu}+\mathbf{b}X_{t}\right)\mathbb{E}\left[f_{t}\left(\ln\left(S_{t^{-}}\right)+J_t,X_t+\mathbf{e}\right)-f_{t}\left(\ln\left(S_{t^{-}}\right),X_t\right)\left|\mathcal{F}_{t^{-}}\right.\right]\mbox{d}t$ in the right-hand side (rhs) of \eqref{term_add_sub}, we obtain:  
\begin{eqnarray*}
\mbox{d}f_t\left(\ln\left(S_t\right),X_t\right) &=& \left[\frac{\partial f_t\left(\ln\left(S_{t^{-}}\right),X_t\right)}{\partial t} + \frac{\partial f_t\left(\ln\left(S_{t^{-}}\right),X_t\right)}{\partial \ln\left(S_t\right)}\alpha_t + \left[\nabla_{X_t}f_t\left(\ln\left(S_{t^{-}}\right),X_t\right)\right]^{\top}\mathbf{A}X_{t}\right]\mbox{d}t\\
&+&\left(\mathbf{\mu}+\mathbf{b}X_{t}\right)\mathbb{E}\left[f_{t}\left(\ln\left(S_{t^{-}}\right)+J_t,X_t+\mathbf{e}\right)-f_{t}\left(\ln\left(S_{t^{-}}\right),X_t\right)\left|\mathcal{F}_{t^{-}}\right.\right]\mbox{d}t+\mbox{d}\bar{M}_t
\end{eqnarray*}
where $\bar{M}_t$ is a process with dynamics
\begin{equation*}
\mbox{d}\bar{M}_t = \left[f_{t}\left(\ln\left(S_{t^{-}}\right)+J_t,X_t+\mathbf{e}\right)-f_{t}\left(\ln\left(S_{t^{-}}\right),X_t\right)\right]\mbox{d}N_t-\left(\mathbf{\mu}+\mathbf{b}^{\top}X_{t}\right)\mathbb{E}\left[f_{t}\left(\ln\left(S_{t^{-}}\right)+J_t,X_{t}+\mathbf{e}\right)-f_{t}\left(\ln\left(S_{t^{-}}\right),X_t\right)\left|\mathcal{F}_{t^{-}}\right.\right]\mbox{d}t.
\end{equation*}
The martingale condition for the joint characteristic function implies that:
\begin{eqnarray}
0&=&\frac{\partial f_t\left(\ln\left(S_{t^{-}}\right),X_t\right)}{\partial t} + \frac{\partial f_t\left(\ln\left(S_{t^{-}}\right),X_t\right)}{\partial \ln\left(S_t\right)}\alpha_t + \left[\nabla_{X_t}f_t\left(\ln\left(S_{t^{-}}\right),X_t\right)\right]^{\top}\mathbf{A}X_{t}\nonumber\\
&+&\left(\mathbf{\mu}+\mathbf{b}X_{t}\right)\mathbb{E}\left[f_{t}\left(\ln\left(S_{t^{-}}\right)+J_t,X_{t}+\mathbf{e}\right)-f_{t}\left(\ln\left(S_{t^{-}}\right),X_{t}\right)\left|\mathcal{F}_{t^{-}}\right.\right].
\label{eq:jcf}
\end{eqnarray}
Under the assumption that the conditional joint characteristic function in \eqref{eq:jointCharLnST} has a log linear affine form as in \eqref{cf_theo},
the partial derivatives in \eqref{eq:jcf} read:
\begin{equation*}
\frac{\partial f_t\left(\ln\left(S_{t^{-}}\right),X_t\right)}{\partial t} = f_t\left(\ln\left(S_{t^{-}}\right),X_t\right) \left[ \frac{\partial u_{0,T}\left(t\right)}{\partial t} + \frac{\partial u_{1,T}\left(t\right)}{\partial t}\ln\left(S_{t^{-}}\right) + \mathbb{J}_{u_{2,T}\left(t\right)}^{\top} X_t \right],
\end{equation*} 
\begin{equation*}
\frac{\partial f_t\left(\ln\left(S_{t^{-}}\right),X_t\right)}{\partial \ln\left(S_t\right)} = f_t\left(\ln\left(S_{t^{-}}\right),X_t\right) u_{1,T}\left(t\right)
\end{equation*}
and
\begin{equation*}
\nabla_{X_t}f_t\left(\ln\left(S_{t^{-}}\right),X_t\right)= f_t\left(\ln\left(S_{t^{-}}\right),X_t\right) u_{2,T}\left(t\right).
\end{equation*}
As the last term in \eqref{eq:jcf} can be expressed in the following way:
\begin{equation*}
\mathbb{E}\left[f_{t}\left(\ln\left(S_{t^{-}}\right)+J_t,X_{t}+\mathbf{e}\right)-f_{t}\left(\ln\left(S_{t^{-}}\right),X_{t}\right)\left|\mathcal{F}_{t^{-}}\right.\right]=f_{t}\left(\ln\left(S_{t^{-}}\right),X_{t}\right)\mathbb{E}\left[\exp\left[u_{1,T}\left(t\right) J_t + u_{2,T}\left(t\right)^{\top}\mathbf{e}\right]-1\left|\mathcal{F}_{t^{-}}\right.\right],
\end{equation*}
the martingale condition for the joint characteristic function becomes
\begin{eqnarray*}
0 &=& \frac{\partial u_{0,T}\left(t\right)}{\partial t}+ \frac{\partial u_{1,T}\left(t\right)}{\partial t}\ln\left(S_{t^{-}}\right)+\mathbb{J}_{u_{2,T}\left(t\right)}^{\top}X_t + \alpha_t u_{1,T}\left(t\right) \nonumber\\
&+& u_{2,T}\left(t\right)^{\top}\mathbf{A}X_t + \left(\mathbf{\mu}+\mathbf{b}^{\top}X_{t}\right) \mathbb{E}\left[\exp\left[u_{1,T}\left(t\right) J_t + u_{2,T}\left(t\right)^{\top}\mathbf{e}\right]-1\left|\mathcal{F}_{t^{-}}\right.\right].
\end{eqnarray*}
Rearranging, we have:
\begin{eqnarray*}
0 &=& \frac{\partial u_{0,T}\left(t\right)}{\partial t}+ \alpha_t u_{1,T}\left(t\right)  + \mathbf{\mu}\mathbb{E}\left[\exp\left[u_{1,T}\left(t\right) J_t + u_{2,T}\left(t\right)^{\top}\mathbf{e}\right]-1\left|\mathcal{F}_{t^{-}}\right.\right] + \frac{\partial u_{1,T}\left(t\right)}{\partial t}\ln\left(S_{t^{-}}\right)  \nonumber\\
&+& \mathbb{J}_{u_{2,T}\left(t\right)}^{\top}X_t + u_{2,T}\left(t\right)^{\top}\mathbf{A}X_t + \mathbb{E}\left[\exp\left[u_{1,T}\left(t\right) J_t + u_{2,T}\left(t\right)^{\top}\mathbf{e}\right]-1\left|\mathcal{F}_{t^{-}}\right.\right] \mathbf{b}^{\top}X_{t} .
\end{eqnarray*}
Recalling that $\alpha_t$ as in \eqref{condAlphat} with $\lambda_t$ as in \eqref{Def:Lambda}, we get:
\begin{eqnarray}
0 &=& \frac{\partial u_{0,T}\left(t\right)}{\partial t}+ \left(r+\varphi\mu\right) u_{1,T}\left(t\right)  + \mathbf{\mu}\mathbb{E}\left[\exp\left[u_{1,T}\left(t\right) J_t + u_{2,T}\left(t\right)^{\top}\mathbf{e}\right]-1\left|\mathcal{F}_{t^{-}}\right.\right] + \frac{\partial u_{1,T}\left(t\right)}{\partial t}\ln\left(S_{t^{-}}\right)  \nonumber\\
&+& \mathbb{J}_{u_{2,T}\left(t\right)}^{\top}X_t + \varphi u_{1,T}\left(t\right)\mathbf{b}^{\top}X_t  + u_{2,T}\left(t\right)^{\top}\mathbf{A}X_t + \mathbb{E}\left[\exp\left[u_{1,T}\left(t\right) J_t + u_{2,T}\left(t\right)^{\top}\mathbf{e}\right]-1\left|\mathcal{F}_{t^{-}}\right.\right] \mathbf{b}^{\top}X_{t}.
\label{res2logpriceP}
\end{eqnarray}
Equation \eqref{res2logpriceP} leads to the following system of ordinary differential equations:
\begin{equation}
\left\{
\begin{array}{l}
\frac{\partial u_{0,T}\left(t\right)}{\partial t} = \mathbf{\mu}\mathbb{E}\left[1 - \exp\left[u_{1,T}\left(t\right) J_t + u_{2,T}\left(t\right)^{\top}\mathbf{e}\right]\left|\mathcal{F}_{t^{-}}\right.\right] - \left(r+\varphi\mu\right) u_{1,T}\left(t\right)\\
\frac{\partial u_{1,T}\left(t\right)}{\partial t} = 0\\
\mathbb{J}_{u_{2,T}\left(t\right)}^{\top} = \mathbb{E}\left[1-\exp\left[u_{1,T}\left(t\right) J_t + u_{2,T}\left(t\right)^{\top}\mathbf{e}\right]\left|\mathcal{F}_{t^{-}}\right.\right] \mathbf{b}^{\top} - u_{2,T}\left(t\right)^{\top}\mathbf{A} -\varphi u_{1,T}\left(t\right)\mathbf{b}^{\top}
\end{array}
\right.,
\label{b000AA}
\end{equation}
with the final conditions that arise by substituting $t = T$ in \eqref{eq:jointCharLnST}:
\begin{equation}
\left\{
\begin{array}{l}
u_{0,T}\left(T\right)=0\\
u_{1,T}\left(T\right)= i u_1\\
u_{2,T}\left(T\right)= i u_2\\
\end{array}
\right. .
\label{eq:FinalCond}
\end{equation}
Notice that, since $\frac{\partial u_{1,T}\left(t\right)}{\partial t} = 0$, we have $u_{1,T}\left(t\right)=i u_1 \ \forall t\leq T$. Exploiting the i.i.d. assumption of the sequence $\left\{J_{k}\right\}_{k\geq 1}$, system \eqref{b000AA} can be rewritten as:
\begin{equation}
\left\{
\begin{array}{l}
\frac{\partial u_{0,T}\left(t\right)}{\partial t} = \mathbf{\mu}\left[1 - \mathbb{E}\left(e^{i u_1 \bar{J}}\right)e^{u_{2,T}\left(t\right)^{\top}\mathbf{e}}\right] - i u_1\left(r+\varphi\mu\right) \\
\mathbb{J}_{u_{2,T}\left(t\right)}^{\top} = \left[1 - \mathbb{E}\left(e^{i u_1 \bar{J}}\right)e^{u_{2,T}\left(t\right)^{\top}\mathbf{e}}\right]\mathbf{b}^{\top} - u_{2,T}\left(t\right)^{\top}\mathbf{A} -i u_1 \varphi \mathbf{b}^{\top}
\end{array}
\right.,
\label{b00}
\end{equation}
with the final conditions in \eqref{eq:FinalCond}. System \eqref{b01} is obtained by choosing $u_2=\mathbf{0}$ and the final conditions in \eqref{eq:FinalCond} write as in \eqref{finalCondLogPrice}.
\end{proof}




\subsection{Change of measure}

The model presented in the previous section describes an incomplete market. Consequently, it is essential to introduce a Radon-Nikodym derivative in order to ensure that, under the risk-neutral probability measure $\mathbb{Q}$, the discounted price is a (local) martingale process. \newline Let $\left\{\theta_{t}\right\}_{t\in\left[0,\mathbb{T}\right]}$ be a predictable process belonging to the domain of the moment generating function (mgf) of $\bar{J}$, i.e. $\theta_{t} \in \left\{c \in \mathbb{R} : \mathbb{E}\left[e^{c\bar{J}}\right]<+\infty\right\}$ $\forall t\in\left[0, \ \mathbb{T}\right]$. We consider the Radon-Nikodym derivative $\left\{L_t\right\}_{t\in\left[0,\mathbb{T}\right]}$ as defined by the following SDE:
\begin{equation}
\mbox{d}L_{t}=L_{t^{-}}\left(\frac{e^{\theta_t J_t}}{\mathbb{E}\left[e^{\theta_t J_t}\left|\mathcal{F}_{t^{-}}\right.\right]}-1\right)\mbox{d}N_{t}, \text{ with } L_0 = 1.
\label{eq:L}
\end{equation}
Using the Dol\'eans-Dade exponential, the solution of the SDE \eqref{eq:L} for all maturities $T \in \left[0,\mathbb{T}\right]$ reads:
\begin{equation}
L_T = \exp \left\{\int_{0}^{T} \left[\theta_{t}J_{t}-\ln\left(\mathbb{E}\left(e^{\theta_{t}J_t}\left|\mathcal{F}_{t^{-}}\right.\right)\right)\right]\mbox{d}N_t\right\}. 
\label{LT}
\end{equation}
Now we need to find a measure $\mathbb{Q}$ such that the martingale condition for the discounted asset price is satisfied $\forall t_0, T \in \left[0,\mathbb{T}\right]$ 
with $t_0<T$,  i.e.:
\begin{equation}
\mathbb{E}^{\mathbb{Q}}\left[e^{-r T}S_{T}\left|\mathcal{F}_{t_0}\right.\right]=e^{-r t_0}S_{t_0}.
\label{eq:QCond}
\end{equation}
Applying the Bayes' theorem \cite{bjork2009arbitrage}, \eqref{eq:QCond} becomes:
\begin{equation}
\mathbb{E}^{\mathbb{Q}}\left[e^{-r T}S_{T}\left|\mathcal{F}_{t_0}\right.\right]=\mathbb{E}\left[e^{-r T}S_{T}\frac{L_{T}}{L_{t_0}}\left|\mathcal{F}_{t_0}\right.\right]= e^{-r t_0}S_{t_0},
\label{eq:underQPrice}
\end{equation}
implying that $\left\{e^{-rt}S_t L_t\right\}_{t\in\left[0,\mathbb{T}\right]}$ is a $\mathbb{P}$-martingale.\newline
Let $\tilde{Z}_{t}:=e^{-r t}S_{t}L_{t}$, we obtain:
\begin{equation}
\mbox{d}\tilde{Z}_{t}=-r \tilde{Z}_{t} \mbox{d}t + e^{-r t}\mbox{d}\left(S_{t}L_{t}\right).
\label{eq:ZetaTilde}
\end{equation} 
Focusing on the term $S_{t}L_{t}$, we have:
\begin{eqnarray}
\mbox{d}\left(S_{t}L_{t}\right) &=& S_{t^{-}}\mbox{d}L_{t}+L_{t^{-}}\mbox{d}S_{t}+ \Delta L_{t} \Delta S_{t}\nonumber\\
&=& S_{t^{-}}\mbox{d}L_{t}+L_{t^{-}}S_{t^{-}}\left[\alpha_t \mbox{d}t+\left(e^{J_t}-1\right)\mbox{d}N_t\right]+S_{t^{-}}L_{t^{-}}\left(e^{J_t}-1\right)\left[\frac{e^{\theta_t J_t}}{\mathbb{E}\left[e^{\theta_t J_t}\left|\mathcal{F}_{t^{-}}\right.\right]}-1\right]\mbox{d}N_t\nonumber\\
&=& S_{t^{-}}\mbox{d}L_{t}+L_{t^{-}}S_{t^{-}}\alpha_t \mbox{d}t+S_{t^{-}}L_{t^{-}}\left[\frac{e^{\left(\theta_t+1\right) J_t}-e^{\theta_t J_t}}{\mathbb{E}\left[e^{\theta_t J_t}\left|\mathcal{F}_{t^{-}}\right.\right]}\right]\mbox{d}N_t.     
\label{eq:intermstep2}
\end{eqnarray}
Defining $M_{t}$ and $\tilde{M}_{t}$ as follows
\begin{eqnarray*}
\mbox{d}M_{t} &:=& \left(\frac{e^{\theta_{t}J_{t}}}{\mathbb{E}\left[e^{\theta_{t}J_t}\left|\mathcal{F}_{t^{-}}\right.\right]}-1\right)\mbox{d}N_{t}\\
\mbox{d}\tilde{M}_{t} &:=& \left(\frac{e^{\left(\theta_{t}+1\right)J_{t}}-e^{\theta_{t}J_{t}}}{\mathbb{E}\left[e^{\theta_{t}J_t}\left|\mathcal{F}_{t^{-}}\right.\right]}\right)\mbox{d}N_{t}-\lambda_t\left(\frac{\mathbb{E}\left[e^{\left(\theta_{t}+1\right)J_{t}}\left|\mathcal{F}_{t^{-}}\right.\right]}{\mathbb{E}\left[e^{\theta_{t}J_t}\left|\mathcal{F}_{t^{-}}\right.\right]}-1\right)\mbox{d}t,
\end{eqnarray*}
the stochastic differential equation in \eqref{eq:intermstep2} becomes:
\begin{equation}
\mbox{d}\left(S_{t}L_{t}\right)= L_{t^{-}}S_{t^{-}}\left[\alpha_t + \lambda_t\left(\frac{\mathbb{E}\left[e^{\left(\theta_{t}+1\right)J_{t}}\left|\mathcal{F}_{t^{-}}\right.\right]}{\mathbb{E}\left[e^{\theta_{t}J_t}\left|\mathcal{F}_{t^{-}}\right.\right]}-1\right)\right] \mbox{d}t+ L_{t^{-}}S_{t^{-}}\left[\mbox{d}M_{t}+\mbox{d}\tilde{M}_{t}\right].
\label{eq:LSandLocalMart}
\end{equation}
Combining \eqref{eq:LSandLocalMart} with \eqref{eq:ZetaTilde}, we obtain:
\begin{equation*}
\mbox{d}\tilde{Z}_{t}=\left[\alpha_t-r + \lambda_t\left(\frac{\mathbb{E}\left[e^{\left(\theta_{t}+1\right)J_{t}}\left|\mathcal{F}_{t^{-}}\right.\right]}{\mathbb{E}\left[e^{\theta_{t}J_t}\left|\mathcal{F}_{t^{-}}\right.\right]}-1\right)\right] \tilde{Z}_{t} \mbox{d}t + \tilde{Z}_{t^{-}}\left[\mbox{d}M_{t}+\mbox{d}\tilde{M}_{t}\right].
\end{equation*}
Note that the requirement of zero drift for the process $\tilde{Z}_t$ implies:
\begin{equation}
r-\alpha_t = \lambda_t\left(\frac{\mathbb{E}\left[e^{\left(\theta_{t}+1\right)J_{t}}\left|\mathcal{F}_{t^{-}}\right.\right]}{\mathbb{E}\left[e^{\theta_{t}J_t}\left|\mathcal{F}_{t^{-}}\right.\right]}-1\right).
\label{eq:LSandLocalMart0}
\end{equation}
Suppose $\alpha_t$ to be specified as in \eqref{condAlphat}, then \eqref{eq:LSandLocalMart0} can be equivalently expressed as
\begin{equation}
1-\varphi = \frac{\mathbb{E}\left[e^{\left(\theta_{t}+1\right)J_{t}}\left|\mathcal{F}_{t^{-}}\right.\right]}{\mathbb{E}\left[e^{\theta_{t}J_t}\left|\mathcal{F}_{t^{-}}\right.\right]}.
\label{eq:LSandLocalMart1}
\end{equation}
If equation \eqref{eq:LSandLocalMart1} admits a solution with respect to the variable $\theta_t$, then this solution is independent of time. That is to say, $\forall t \in \left[0,\ \mathbb{T}\right]$, we have that $\theta_{t}=\theta^{\star}$. For example, if $\bar{J}\sim \mathcal{N}\left(\mu_J,\sigma^2_{J}\right)$, then $\theta^{\star} =\frac{\ln\left(1-\varphi\right)-\mu_J}{\sigma^2_J}-\frac12$ with $\varphi<1$.

Let $\left\{L_{t}^{\star}\right\}_{t\in\left[0, \ \mathbb{T}\right]}$ be the Radon-Nikodym derivative in \eqref{LT} at $\theta^{\star}$ solution of \eqref{eq:LSandLocalMart1} and suppose that $\sup_{t\in\left[0,\mathbb{T}\right]}\mathbb{E}\left[ L_{t}^{\star}\right]<+\infty$. Then $\left\{L_{t}^{\star}\right\}_{t\in\left[0, \ \mathbb{T}\right]}$ is a strictly positive martingale over the interval $\left[0, \ \mathbb{T}\right]$. Indeed, we have:
\begin{eqnarray*}
\mathbb{E}\left[L_{T}^{\star}\left|\mathcal{F}_{t_0}\right.\right]&=&L_{t_0}^{\star}\mathbb{E}\left\{\exp\left[\int_{t_0}^{T}\left(\theta^{\star}J_t-\ln\left(\mathbb{E}\left[e^{\theta^{\star}\bar{J}}\right]\right)\right)\mbox{d}N_{t}\right]\left|\mathcal{F}_{t_0}\right.\right\}\nonumber\\
&=&L_{t_0}^{\star}\mathbb{E}\left\{\exp\left[\sum_{k=N_{t_0}+1}^{N_{T}}\left(\theta^{\star}J_k-\ln\left(\mathbb{E}\left[e^{\theta^{\star}\bar{J}}\right]\right)\right)\right]\left|\mathcal{F}_{t_0}\right.\right\}.
\end{eqnarray*}
Consider the $\sigma$-field $\mathcal{G}_{t_0}^{T}:=\sigma\left(\mathcal{F}_{t_0} \cup \sigma\left(N_{T}\right)\right)$. Through the tower rule we obtain: 
\begin{eqnarray*}
\mathbb{E}\left[L_{T}^{\star}\left|\mathcal{F}_{t_0}\right.\right]&=&L_{t_0}^{\star}\mathbb{E}\left[\mathbb{E}\left[\prod_{k=N_{t_0}+1}^{N_T} \frac{e^{\theta^{\star}J_k}}{\mathbb{E}\left(e^{\theta^{\star}\bar{J}}\right)} \left|\mathcal{G}^T_{t_0}\right.\right]\left|\mathcal{F}_{t_0}\right.\right].
\end{eqnarray*}
Exploiting the i.i.d. assumption for the jump size and its independence from the counting process $N_T$, we get:
\begin{equation*}
\mathbb{E}\left[L_{T}^{\star}\left|\mathcal{F}_{t_0}\right.\right]=L_{t_0}^{\star}\mathbb{E}\left[\prod_{k=N_{t_0}+1}^{N_T} \mathbb{E}\left[ \frac{e^{\theta^{\star}J_k}}{\mathbb{E}\left(e^{\theta^{\star}\bar{J}}\right)} \right]\left|\mathcal{F}_{t_0}\right.\right] = L_{t_0}^{\star}.
\end{equation*}
%
 %

\subsection{Asset price dynamics under the risk-neutral measure}


The Radon-Nikodym derivative, determined in the previous section, makes the discounted underlying asset price a (local) martingale. In this section we present the characteristic function of the log-price under the risk-neutral measure $\mathbb{Q}$.  

\begin{theorem}\label{cf_Q}
The characteristic function of the log-price under $\mathbb{Q}$ measure reads
	\begin{equation}
		\mathbb{E}^{\mathbb{Q}}\left[e^{iu\ln\left(S_{T}\right)}\left|\mathcal{F}_{t_0}\right.\right]=S_{t_0}^{iu}\exp\left[u_{0,T}\left(t_0\right)+u_{2,T}\left(t_0\right)^{\top}X_{t_0}\right],
	\end{equation}
where the time coefficients $u_{0,T}\left(\cdot\right)$ and $u_{2,T}\left(\cdot\right)$ satisfy the following system of ordinary differential equations 
\begin{equation}
	\left\{\begin{array}{l}
		\frac{\partial u_{0,T}\left(t\right)}{\partial t}= \mu \left[1-\mathbb{E}^{\mathbb{Q}}\left(e^{i u \bar{J}^{\mathbb{Q}}}\right) e^{u_{2,T}\left(t\right)^{\top}\mathbf{e}}\right]-  i u \left\{r-\mu\left[\mathbb{E}^{\mathbb{Q}}\left(e^{\bar{J}^{\mathbb{Q}}}\right) -1\right]\right\} \\
		\mathbb{J}_{u_{2,T}\left(t\right)}^{\top} = \left[1-\mathbb{E}^{\mathbb{Q}}\left(e^{i u \bar{J}^{\mathbb{Q}}} \right) e^{u_{2,T}\left(t\right)^{\top}\mathbf{e}}\right]\mathbf{b}^{\top} - u_{2,T}\left(t\right)^{\top}\mathbf{A} +i u \left[\mathbb{E}^{\mathbb{Q}}\left(e^{\bar{J}^{\mathbb{Q}}}\right) -1\right] \mathbf{b}^{\top},
	\end{array}
	\right.
	\label{MainResultForPricing}
\end{equation}	
with final conditions
\begin{equation}
	\left\{
	\begin{array}{l}
		u_{0,T}\left(T\right)=0\\
		u_{2,T}\left(T\right)= \mathbf{0}\\
	\end{array}
	\right. .
	\label{finalCondLogPrice_Q}
\end{equation}
The density of the jump size under the risk-neutral measure is obtained using the Esscher transform as follows
\begin{equation}
f_{\bar{J}^{\mathbb{Q}}}\left(j\right) =f_{\bar{J}}\left(j\right)\frac{e^{\theta^{\star}j}}{\mathbb{E}\left[e^{\theta^{\star}\bar{J}}\right]}.
\label{EsschJdens}
\end{equation}
\end{theorem}

\begin{proof}
To determine the characteristic function of the log-price under $\mathbb{Q}$ we need the Radon-Nikodym derivative $\left\{L_t^{\star}\right\}_{t\in\left[0,\mathbb{T}\right]}$ as:\footnote{We remind that $L_t^{\star}$ is the process defined in \eqref{LT} at $\theta^{\star}$ solution of \eqref{eq:LSandLocalMart1} with the assumption $\sup_{t\in\left[0,\mathbb{T}\right]}\mathbb{E}\left[ L_{t}^{\star}\right]<+\infty$.}
\begin{equation}
	\mathbb{E}^{\mathbb{Q}}\left[e^{iu\ln\left(S_{T}\right)}\left|\mathcal{F}_{t_0}\right.\right]= \mathbb{E}\left[e^{iu\ln\left(S_{T}\right)}\frac{L_{T}^{\star}}{L_{t_0}^{\star}}\left|\mathcal{F}_{t_0}\right.\right].
\end{equation}
We define 
\begin{equation}
	\mbox{d}m_t := \mbox{d} \ln\left(L_{t}^{\star}\right)  = \left\{\theta^{\star}J_t-\ln\left[\mathbb{E}\left(e^{\theta^{\star}\bar{J}}\right)\right]\right\}\mbox{d}N_t,
	\label{baba}
\end{equation} 
and, $\forall t \in \left[t_0, \ T\right]$, we consider the following $\mathbb{P}$-martingale process
\begin{equation}
	f_{t}\left(\ln\left(S_{t}\right),X_{t},m_{t}\right):=\mathbb{E}^{\mathbb{P}}\left[e^{iu\ln\left(S_{T}\right)}L_{T}^{\star}\left|\mathcal{F}_{t}\right.\right].
	\label{aaaaaaa}
\end{equation}

\noindent Applying It\^o's lemma to \eqref{aaaaaaa}, we obtain
\begin{eqnarray}
	\mbox{d}\ft &=& \frac{\partial \ftm}{\partial t}\mbox{d}t+ \frac{\partial \ftm}{\partial \ln\left(S_t\right)}\mbox{d}\ln\left(S_t\right)+\nabla_{X_t}\ftm^{\top}\mbox{d}X_t\nonumber\\
	&+&\frac{\partial \ftm}{\partial m_t}\mbox{d}m_t+\left[\ftj-f_{t}\left(\ln\left(S_{t^{-}}\right),X_{t},m_{t^{-}}\right)\right]\mbox{d}N_{t}\nonumber\\
	&-&\left[\frac{\partial \ftm}{\partial \ln\left(S_t\right)}\Delta \ln\left(S_t\right) +\frac{\partial \ftm}{\partial m_t}\Delta m_t   \right]\mbox{d}N_{t}\nonumber\\
	&-&\nabla_{X_t}\ftm^{\top} \Delta X_t \mbox{d}N_{t}.
	\label{eq:ft1}
\end{eqnarray}
From \eqref{pureJumpwithnounitary} we get $\mbox{d}\ln \left(S_t\right)$ and, inserting \eqref{eq:state_process} and \eqref{baba} into \eqref{eq:ft1}, we obtain:
\begin{eqnarray*}
	\mbox{d}\ft &=& \left\{\frac{\partial \ftm}{\partial t}+\alpha_t\frac{\partial \ftm}{\partial \ln\left(S_t\right)} + \nabla_{X_t}\ftm^{\top}\mathbf{A}X_{t}\right\}\mbox{d}t\nonumber\\
	&+&\left[\ftj-f_{t}\left(\ln\left(S_{t^{-}}\right),X_{t},m_{t^{-}}\right)\right]\mbox{d}N_{t}\nonumber\\
	&=& \left\{\frac{\partial \ftm}{\partial t}+\alpha_t\frac{\partial \ftm}{\partial \ln\left(S_t\right)}+\nabla_{X_t}\ftm^{\top}\mathbf{A}X_{t}\right\}\mbox{d}t\nonumber\\
	&+&\left(\mu+\mathbf{b}^{\top}X_{t}\right)\mathbb{E}\left[\ftj-f_{t}\left(\ln\left(S_{t^{-}}\right),X_{t},m_{t^{-}}\right)\left|\mathcal{F}_{t^{-}}\right.\right]\mbox{d}t\nonumber\\
	&+&\left\{\left[\ftj-f_{t}\left(\ln\left(S_{t^{-}}\right),X_{t},m_{t^{-}}\right)\right]\mbox{d}N_{t}\right.\nonumber\\
	&-&\left.\left(\mu+\mathbf{b}^{\top}X_{t}\right)\mathbb{E}\left[\ftj-f_{t}\left(\ln\left(S_{t^{-}}\right),X_{t},m_{t^{-}}\right)\left|\mathcal{F}_{t^{-}}\right.\right]\right\}\mbox{d}t.
\end{eqnarray*}
The zero drift condition for the process $\left\{f_{t}\left(\ln\left(S_{t}\right),X_{t},m_{t}\right)\right\}_{t \in \left[t_0, \ T\right]}$ in \eqref{aaaaaaa}  becomes:
\begin{eqnarray}
	0 &=& \frac{\partial \ftm}{\partial t}+\alpha_t\frac{\partial \ftm}{\partial \ln\left(S_t\right)} + \nabla_{X_t}\ftm^{\top}\mathbf{A}X_{t}\nonumber\\
	&+&\left(\mu+\mathbf{b}^{\top}X_{t}\right)\mathbb{E}\left[\ftj-f_{t}\left(\ln\left(S_{t^{-}}\right),X_{t},m_{t^{-}}\right)\left|\mathcal{F}_{t^{-}}\right.\right]\nonumber\\
	&=& \frac{\partial \ftm}{\partial t}+\alpha_t\frac{\partial \ftm}{\partial \ln\left(S_t\right)} +\mu \mathbb{E}\left[\ftj-f_{t}\left(\ln\left(S_{t^{-}}\right),X_{t},m_{t^{-}}\right)\left|\mathcal{F}_{t^{-}}\right.\right] \nonumber\\
	&+&\left\{\nabla_{X_t}\ftm^{\top}\mathbf{A}+\mathbb{E}\left[\ftj-f_{t}\left(\ln\left(S_{t^{-}}\right),X_{t},m_{t^{-}}\right)\left|\mathcal{F}_{t^{-}}\right.\right]\mathbf{b}^{\top}\right\}X_{t}.
\end{eqnarray}
We assume a log-affine linear form for the function $\ftzero$ that reads:
\begin{equation*}
	\ftzero = \exp\left[u_{0,T}\left(t_0\right)+u_{1,T}\left(t_0\right)\ln\left(S_{t_0}\right)+u_{2,T}\left(t_0\right)^{\top}X_{t_0}+u_{3,T}\left(t_0\right)m_{t_0}\right].
\end{equation*}
We compute the following partial derivatives
\begin{eqnarray*}
	\frac{\partial \ftm}{\partial t}&=&\ftm\left[\frac{\partial u_{0,T}\left(t\right)}{\partial t}+\frac{\partial u_{1,T}\left(t\right)}{\partial t}\ln\left(S_{t^{-}}\right) + \mathbb{J}_{u_{2,T}\left(t\right)}^{\top}X_{t} + \frac{\partial u_{3,T}\left(t\right)}{\partial t}m_{t^{-}} \right]\nonumber\\ \\
	\frac{\partial \ftm}{\partial \ln\left(S_t\right)}&=& \ftm u_{1,T}\left(t\right)\nonumber\\\\
	\nabla_{X_{t}}\ftm &=& \ftm u_{2,T}\left(t\right)\nonumber\\\\
	\frac{\partial \ftm}{\partial m_t} &=& \ftm u_{3,T}\left(t\right),\nonumber\\ \\
\end{eqnarray*}
and we observe that
\begin{eqnarray*}
	\mathbb{E}\left[\ftj - f_{t}\left(\ln\left(S_{t^{-}}\right),X_{t},m_{t^{-}}\right)\left|\mathcal{F}_{t^{-}}\right.\right] &=&\ftm \mathbb{E}\left\{ \exp\left[ \left[u_{1,T}\left(t\right) +   u_{3,T}\left(t\right) \theta^{\star}\right] J_t \right.\right.\nonumber\\
	&+& \left.\left. u_{2,T}\left(t\right)^{\top}\mathbf{e} - u_{3,T}\left(t\right)\ln\left[\mathbb{E}\left(e^{\theta^{\star}\bar{J}}\right) \right]\right]-1\left|\mathcal{F}_{t^{-}}\right.\right\}\nonumber\\
	&=& \left\{\mathbb{E}\left[\frac{e^{\left(u_{1,T}\left(t\right) +   u_{3,T}\left(t\right) \theta^{\star}\right) J_t}}{\left[\mathbb{E}\left(e^{\theta^{\star}\bar{J}} \right)\right]^{u_{3,T}\left(t\right)}} \left|\mathcal{F}_{t^{-}}\right.\right] e^{u_{2,T}\left(t\right)^{\top}\mathbf{e}}-1\right\} \nonumber\\
	&\times& \ftm.
\end{eqnarray*}
Therefore, we have:
\begin{eqnarray*}
	0 &=& \frac{\partial u_{0,T}\left(t\right)}{\partial t}+\frac{\partial u_{1,T}\left(t\right)}{\partial t}\ln\left(S_{t^{-}}\right) + \mathbb{J}_{u_{2,T}\left(t\right)}^{\top}X_t + \frac{\partial u_{3,T}\left(t\right)}{\partial t}m_{t^{-}} \nonumber\\
	&+&\alpha_t u_{1,T}\left(t\right) +\mu \left\{\mathbb{E}\left[\frac{e^{\left(u_{1,T}\left(t\right) +   u_{3,T}\left(t\right) \theta^{\star}\right) J_t}}{\left[\mathbb{E}\left(e^{\theta^{\star}\bar{J}}\right)\right]^{u_{3,T}\left(t\right)}} \left|\mathcal{F}_{t^{-}}\right.\right] e^{u_{2,T}\left(t\right)^{\top}\mathbf{e}}-1\right\} \nonumber\\
	&+&\left\{u_{2,T}\left(t\right)^{\top}\mathbf{A}+\left\{\mathbb{E}\left[\frac{e^{\left(u_{1,T}\left(t\right) +   u_{3,T}\left(t\right) \theta^{\star}\right) J_t}}{\left[\mathbb{E}\left(e^{\theta^{\star}\bar{J}}\right)\right]^{u_{3,T}\left(t\right)}} \left|\mathcal{F}_{t^{-}}\right.\right] e^{u_{2,T}\left(t\right)^{\top}\mathbf{e}}-1\right\}\mathbf{b}^{\top}\right\}X_{t}\nonumber\\
\end{eqnarray*}
and, using \eqref{eq:LSandLocalMart0}, it can be rearranged as follows:
\begin{eqnarray}\label{bicocca}
	0 &=& \frac{\partial u_{1,T}\left(t\right)}{\partial t}\ln\left(S_{t^{-}}\right)  + \frac{\partial u_{3,T}\left(t\right)}{\partial t}m_{t^{-}} + \frac{\partial u_{0,T}\left(t\right)}{\partial t} + \left\{r-\mu\left\{\mathbb{E}\left[\frac{e^{\left(1 +  \theta^{\star}\right) J_t}}{\mathbb{E}\left(e^{\theta^{\star}\bar{J}}\right)} \left|\mathcal{F}_{t^{-}}\right.\right] -1\right\}\right\} u_{1,T}\left(t\right) \nonumber\\
	&+& \mu \left\{\mathbb{E}\left[\frac{e^{\left(u_{1,T}\left(t\right) +   u_{3,T}\left(t\right) \theta^{\star}\right) J_t}}{\left[\mathbb{E}\left(e^{\theta^{\star}\bar{J}}\right)\right]^{u_{3,T}\left(t\right)}} \left|\mathcal{F}_{t^{-}}\right.\right] e^{u_{2,T}\left(t\right)^{\top}\mathbf{e}}-1\right\} \nonumber\\
	&+&\left\{u_{2,T}\left(t\right)^{\top}\mathbf{A}+\left\{\mathbb{E}\left[\frac{e^{\left(u_{1,T}\left(t\right) +   u_{3,T}\left(t\right) \theta^{\star}\right) J_t}}{\left[\mathbb{E}\left(e^{\theta^{\star}\bar{J}}\right)\right]^{u_{3,T}\left(t\right)}} \left|\mathcal{F}_{t^{-}}\right.\right] e^{u_{2,T}\left(t\right)^{\top}\mathbf{e}}-1\right\}\mathbf{b}^{\top}+ \mathbb{J}_{u_{2,T}\left(t\right)}^{\top}\right\}X_{t}\nonumber\\
	&-&\left\{\mathbb{E}\left[\frac{e^{\left(1 +  \theta^{\star}\right) J_t}}{\mathbb{E}\left(e^{\theta^{\star}\bar{J}}\right)} \left|\mathcal{F}_{t^{-}}\right.\right] -1\right\} u_{1,T}\left(t\right)\mathbf{b}^{\top}X_t.
\end{eqnarray}
Exploiting the i.i.d. assumption for the jump size, the quantity $\mathbb{E}\left[\frac{e^{\left(1 +  \theta^{\star}\right) J_t}}{\mathbb{E}\left(e^{\theta^{\star}\bar{J}}\right)} \left|\mathcal{F}_{t^{-}}\right.\right]$ in \eqref{bicocca} can be rewritten as
\begin{equation*}
	\mathbb{E}\left[\frac{e^{\left(1 +  \theta^{\star}\right) J_t}}{\mathbb{E}\left(e^{\theta^{\star}\bar{J}}\right)} \left|\mathcal{F}_{t^{-}}\right.\right] =\frac{\mathbb{E}\left[e^{\left(1 + \theta^{\star}\right)\bar{J}}\right]}{\mathbb{E}\left[e^{\theta^{\star}\bar{J}}\right]}.
\end{equation*}
We note that the time-coefficients $u_{0,T}\left(t\right)$, $u_{1,T}\left(t\right)$, $u_{2,T}\left(t\right)$ and $u_{3,T}\left(t\right)$ satisfy the following system of ordinary differential equations:
\begin{equation}
	\left\{\begin{array}{l}
		\frac{\partial u_{0,T}\left(t\right)}{\partial t}= \mu \left\{1-\mathbb{E}\left[\frac{e^{\left(u_{1,T}\left(t\right) +   u_{3,T}\left(t\right) \theta^{\star}\right) J_t}}{\left[\mathbb{E}\left(e^{\theta^{\star}\bar{J}}\right)\right]^{u_{3,T}\left(t\right)}} \left|\mathcal{F}_{t^{-}}\right.\right] e^{u_{2,T}\left(t\right)^{\top}\mathbf{e}}\right\}- \left\{r-\mu\left\{ \frac{\mathbb{E}\left[e^{\left(1 + \theta^{\star}\right)\bar{J}}\right]}{\mathbb{E}\left[e^{\theta^{\star}\bar{J}}\right]} -1\right\}\right\} u_{1,T}\left(t\right) \\
		\frac{\partial u_{1,T}\left(t\right)}{\partial t} = 0 \\  
		\mathbb{J}_{u_{2,T}\left(t\right)}^{\top} = \left\{1-\mathbb{E}\left[\frac{e^{\left(u_{1,T}\left(t\right) +   u_{3,T}\left(t\right) \theta^{\star}\right) J_t}}{\left[\mathbb{E}\left(e^{\theta^{\star}\bar{J}}\right)\right]^{u_{3,T}\left(t\right)}} \left|\mathcal{F}_{t^{-}}\right.\right] e^{u_{2,T}\left(t\right)^{\top}\mathbf{e}}\right\}\mathbf{b}^{\top} - u_{2,T}\left(t\right)^{\top}\mathbf{A} \\
		\ \ \ \ \ \ \ \ \ \ \ \ \ \ \ \  + \left\{\frac{\mathbb{E}\left[e^{\left(1 + \theta^{\star}\right)\bar{J}}\right]}{\mathbb{E}\left[e^{\theta^{\star}\bar{J}}\right]}-1\right\} u_{1,T}\left(t\right)\mathbf{b}^{\top}\\
		\frac{\partial u_{3,T}\left(t\right)}{\partial t} = 0 
	\end{array}
	\right.
	\label{a1}
\end{equation}
with the final conditions that arise by substituting $t = T$ in \eqref{aaaaaaa}:
\begin{equation*}
	\left\{
	\begin{array}{l}
		u_{0,T}\left(T\right)= 0  \\
		u_{1,T}\left(T\right)= iu \\ 
		u_{2,T}\left(T\right)= 0  \\ 
		u_{3,T}\left(T\right)= 1 
	\end{array}
	\right. .
\end{equation*}
Noticing that $u_{1,T}\left(t\right)$ and $u_{3,T}\left(t\right)$ do not depend on time, we can rewrite the conditions in \eqref{a1} as follows:
\begin{equation*}
	\left\{\begin{array}{l}
		\frac{\partial u_{0,T}\left(t\right)}{\partial t}= \mu \left\{1-\mathbb{E}\left[\frac{e^{\left(i u  +   \theta^{\star}\right) J_t}}{\left[\mathbb{E}\left(e^{\theta^{\star}\bar{J}}\right)\right]} \left|\mathcal{F}_{t^{-}}\right.\right] e^{u_{2,T}\left(t\right)^{\top}\mathbf{e}}\right\}-  i u \left\{r-\mu\left\{\frac{\mathbb{E}\left[e^{\left(1 + \theta^{\star}\right)\bar{J}}\right]}{\mathbb{E}\left[e^{\theta^{\star}\bar{J}}\right]} -1\right\}\right\} \\
		u_{1,T}\left(t\right) = iu \\  
		\mathbb{J}_{u_{2,T}\left(t\right)}^{\top} = \left\{1-\mathbb{E}\left[\frac{e^{\left(i u +  \theta^{\star}\right) J_t}}{\mathbb{E}\left[e^{\theta^{\star}\bar{J}}\right]} \left|\mathcal{F}_{t^{-}}\right.\right] e^{u_{2,T}\left(t\right)^{\top}\mathbf{e}}\right\}\mathbf{b}^{\top} - u_{2,T}\left(t\right)^{\top}\mathbf{A} +i u \left\{\frac{\mathbb{E}\left[e^{\left(1 + \theta^{\star}\right)\bar{J}}\right]}{\mathbb{E}\left[e^{\theta^{\star}\bar{J}}\right]} -1\right\} \mathbf{b}^{\top}\\
		u_{3,T}\left(t\right) = 1 
	\end{array}
	\right. .
\end{equation*}
Exploiting the i.i.d. assumption for the jump size, finally we get:
\begin{equation}
	\left\{\begin{array}{l}
		\frac{\partial u_{0,T}\left(t\right)}{\partial t}= \mu \left\{1-\frac{\mathbb{E}\left[e^{\left(i u  +   \theta^{\star}\right) \bar{J}}\right]}{\left[\mathbb{E}\left(e^{\theta^{\star}\bar{J}}\right)\right]}  e^{u_{2,T}\left(t\right)^{\top}\mathbf{e}}\right\}-  i u \left\{r-\mu\left\{\frac{\mathbb{E}\left[e^{\left(1 + \theta^{\star}\right)\bar{J}}\right]}{\mathbb{E}\left[e^{\theta^{\star}\bar{J}}\right]} -1\right\}\right\} \\
		u_{1,T}\left(t\right) = iu \\  
		\mathbb{J}_{u_{2,T}\left(t\right)}^{\top} = \left\{1-\frac{\mathbb{E}\left[e^{\left(i u  +   \theta^{\star}\right) \bar{J}}\right]}{\left[\mathbb{E}\left(e^{\theta^{\star}\bar{J}}\right)\right]} e^{u_{2,T}\left(t\right)^{\top}\mathbf{e}}\right\}\mathbf{b}^{\top} - u_{2,T}\left(t\right)^{\top}\mathbf{A} +i u \left\{\frac{\mathbb{E}\left[e^{\left(1 + \theta^{\star}\right)\bar{J}}\right]}{\mathbb{E}\left[e^{\theta^{\star}\bar{J}}\right]} -1\right\} \mathbf{b}^{\top}\\
		u_{3,T}\left(t\right) = 1 
	\end{array}
	\right. .
	\label{a2}
\end{equation}
We observe that the quantities $\frac{\mathbb{E}\left[e^{\left(i u  +   \theta^{\star}\right) \bar{J}}\right]}{\left[\mathbb{E}\left(e^{\theta^{\star}\bar{J}}\right)\right]}$ and $\frac{\mathbb{E}\left[e^{\left(1 + \theta^{\star}\right)\bar{J}}\right]}{\mathbb{E}\left[e^{\theta^{\star}\bar{J}}\right]}$ are respectively the characteristic function and the first exponential moment of a random variable $\bar{J}^{\mathbb{Q}}$ whose density is defined in \eqref{EsschJdens}. 
This implies that $\frac{\partial u_{0,T}\left(t\right)}{\partial t}$ and $\mathbb{J}_{u_{2,T}\left(t\right)}^{\top}$ coincide with quantities in \eqref{MainResultForPricing}.

Finally, the conditional characteristic function of the log-price at maturity $T$ under $\mathbb{Q}$ can be written as 
\begin{eqnarray}
	\mathbb{E}^\mathbb{Q}\left[e^{iu \ln S_{T}}\left| \mathcal{F}_{t_0}\right.\right] &=& f_{t_0} \left(\ln S_{t_0}, X_{t_0}, m_{t_0}\right)e^{-m_{t_0}} \nonumber \\ 
	 &=&  S_{t_0}^{iu}\exp\left[u_{0,T}\left(t_0\right)+u_{2,T}\left(t_0\right)^{\top}X_{t_0}\right], 
	\label{cf_fform}
\end{eqnarray}
where the time coefficients $u_{0,T}\left(\cdot\right)$ and $u_{2,T}\left(\cdot\right)$ satisfy the system in \eqref{MainResultForPricing} and this concludes the proof. 
\end{proof}

Posing $\varphi=-\left[\mathbb{E}^{\mathbb{Q}}\left(e^{\bar{J}^{\mathbb{Q}}}\right)-1\right]$ and introducing a sequence $\left\{J_{k}^{\mathbb{Q}}\right\}_{k\geq1}$ of i.i.d. random variables with $J^{\mathbb{Q}}_{k}\overset{d}{=} \bar{J}^{\mathbb{Q}}$, the  conditions in \eqref{b01} and \eqref{MainResultForPricing} meet. This implies that, under $\mathbb{Q}$, the price process $S_{t}$ satisfies the following SDE:
\begin{equation}\label{sde_simplied_version}
	\mbox{d} S_{t}=\left\{r-\left(\mu+\mathbf{b}^{\top}X_t\right)\left[\mathbb{E}^{\mathbb{Q}}\left(e^{\bar{J}^{\mathbb{Q}}}\right)-1\right]\right\}S_{t^{-}}\mbox{d}t+S_{t^{-}}\mbox{d}\tilde{Y}^{\mathbb{Q}}_t,
\end{equation}
where the compound CARMA(p,q)-Hawkes $\tilde{Y}^{\mathbb{Q}}_t$ is defined as:
\begin{equation*}
\tilde{Y}^{\mathbb{Q}}_t:= \sum_{k=1}^{N_t}\left(e^{J_{k}^{\mathbb{Q}}}-1\right) \text{ with } \tilde{Y}^{\mathbb{Q}}_0=0.
\end{equation*}

\begin{remark}
	We present in Appendix~\ref{app1} a simplified version of the model in \eqref{pureJumpwithnounitary} where we assume $\mathbb{E}^{\mathbb{Q}}\left[e^{\bar{J}^\mathbb{Q}}\right] = 1$. In particular, the formula of an European call option turns out to be a weighted sum of call option prices where the weights are the conditional probabilities of events $\left\{N_{T}-N_{t_0}=n\right\}_{n\geq 0}$. We also discuss two numerical methods for the computation of the latter conditional probabilities based on the conditional characteristic function of $\left[X_T,N_T\right]$ derived in the same section.
\end{remark}

The price dynamics in \eqref{pureJumpwithnounitary},  under the physical measure $\mathbb{P}$, can be generalized by adding a diffusion term $\sigma_t S_{t}\mbox{d}W_t$ where  $\left\{\sigma_{t}\right\}_{t\in \left[0,\mathbb{T}\right]}$ is a positive adapted stochastic process and the couple $\left(\sigma_t,W_t\right)$ is independent of the compound CARMA(p,q) process $\tilde{Y}^{\mathbb{Q}}_t$. The resulting Radon-Nikodym derivative for the  price dynamics is the product of \eqref{LT} with the Radon-Nikodym derivative for a Geometric Brownian Motion that is computed by means of the Girsanov Theorem \citep[see eq. 54 page 682 in][]{shiryaev1999essentials}. Under $\mathbb{Q}$, the resulting price dynamics writes
\begin{equation}
	\mbox{d} S_{t}=\left\{r-\left(\mu+\mathbf{b}^{\top}X_t\right)\left(\mathbb{E}^{\mathbb{Q}}\left[e^{\bar{J}^{\mathbb{Q}}}\right]-1\right)\right\}S_{t^{-}}\mbox{d}t+\sigma_t S_{t^{-}}\mbox{d}W_t+S_{t^{-}}\mbox{d}\tilde{Y}^{\mathbb{Q}}_t.
	\label{eq:general_sde}
\end{equation}
The solution of the SDE in \eqref{eq:general_sde} reads
\begin{equation}
	S_{T}= S_{t_0}\exp\left[r\left(T-t_0\right) - \frac{1}{2}\int_{t_0}^{T}\sigma_t^2\mbox{d}t +  \int_{t_0}^{T}\sigma_t\mbox{d}W_t - \left(\mathbb{E}^{\mathbb{Q}}\left[e^{\bar{J}^{\mathbb{Q}}}\right]-1\right)\int_{t_0}^{T}\lambda_t \mbox{d}t + \left(Y^{\mathbb{Q}}_{T}- Y^{\mathbb{Q}}_{t_0}\right)\right]
	\label{eq:solutionS_generalunderQ}
\end{equation}
where $\lambda_t$ is defined in \eqref{Def:Lambda}. The process $\left\{Y^{\mathbb{Q}}_{t}\right\}_{t \in \left[0, \ \mathbb{T}\right]}$ has the same structure as in \eqref{YtDef} with the jump sequence $\left\{J_k\right\}_{k \geq  1}$ substituted with the jump sequence $\left\{J^{\mathbb{Q}}_k\right\}_{k \geq 1}$ introduced in Theorem~\ref{cf_Q}. 
From \eqref{eq:solutionS_generalunderQ} the conditional characteristic function of the log-price at maturity $T$ is: 
\begin{equation}
	\mathbb{E}^{\mathbb{Q}}\left[e^{iu\ln\left(S_{T}\right)}\left|\mathcal{F}_{t_0}\right.\right]=S_{t_0}^{iu}e^{u_{0,T}\left(t_0\right)+u_{2,T}\left(t_0\right)^{\top}X_{t_0}}\mathbb{E}^{\mathbb{Q}}\left[e^{iu\left(\int_{t_0}^{T}\sigma_t\mbox{d}W_t-\frac{1}{2}\int_{t_0}^{T}\sigma_t^2\mbox{d}t\right)} \left|\mathcal{F}_{t_0}\right. \right],
	\label{eq:FG}
\end{equation} 
where the time-dependent coefficients $u_{0,T}\left(t\right)$ and $u_{2,T}\left(t\right)$ are determined in Theorem \eqref{cf_Q}. A pricing formula based on the conditional characteristic function of the log-price can be obtained if the expectation in the right hand side of \eqref{eq:FG} has a log-affine form wrt $\sigma^2_{t_0}$. 

For $\sigma_{t} = \sigma>0$, the price dynamics in \eqref{eq:general_sde} reads
\begin{equation}
	\mbox{d} S_{t}=\left\{r-\left(\mu+\mathbf{b}^{\top}X_t\right)\left(\mathbb{E}^{\mathbb{Q}}\left[e^{\bar{J}^{\mathbb{Q}}}\right]-1\right)\right\}S_{t^{-}}\mbox{d}t+\sigma S_{t^{-}}\mbox{d}W_t+S_{t^{-}}\mbox{d}\tilde{Y}^{\mathbb{Q}}_t,
	\label{eq:STQ}
\end{equation}
and the conditional characteristic function in \eqref{eq:FG} becomes
\begin{equation}\label{cf_general}
	\mathbb{E}^{\mathbb{Q}}\left[e^{iu\ln\left(S_{T}\right)}\left|\mathcal{F}_{t_0}\right.\right]=S_{t_0}^{iu}\exp\left[-iu\frac12\sigma^2\left(T-t_0\right)-\frac{u^2}{2}\sigma^2\left(T-t_0\right) + u_{0,T}\left(t_0\right)+u_{2,T}\left(t_0\right)^{\top}X_{t_0}\right].
\end{equation}

In this paper we emphasize the role of previous jumps in the asset price dynamics. We will present the numerical and empirical analysis in a framework where we assume $\sigma$ to be constant, that is a CARMA(p,q)-Hawkes jump-diffusion framework. The analysis can be extended to the general case as in \eqref{eq:general_sde} specifying the dynamics of the process $\left\{\sigma_t\right\}_{t \in \left[0, \ \mathbb{T}\right]}$.

%% file: GLmethod.tex
The characteristic function of the log-price in \eqref{cf_general} can be used to evaluate European option prices.
In Proposition~\ref{theo_gl} we introduce a numerical pricing formula for European put options.
\begin{proposition}\label{theo_gl}
Let $K$ be the strike price of a European put option and assume that the price dynamics under $\mathbb{Q}$ satisfies the SDE in \eqref{eq:STQ} with characteristic function $\phi(u) := \cflp$ as in \eqref{cf_general}. The European put option price $p\left(K, t_0, T\right)$ at time $t_0$  can be evaluated as follows:
\begin{equation}\label{put_numerical}
	p\left(K, t_0, T\right) \approx e^{-r\left(T - t_0\right)} K\sum_{j=1}^{m}\left[\frac12 - \frac{1}{\pi}\sum_{k=1}^m \mathsf{Re}\left(\frac{e^{-i u_k \left(\ln\left(K\right) - u_{j}\right)}\phi(u_k)}{i u_k}\right)e^{u_k} \omega\left(u_k\right) \right]\omega\left(u_j\right),
\end{equation}
where $\mathsf{Re}\left(\cdot\right)$ returns the real part of the argument, the sequence $\left\{u_h\right\}_{h = 1,\ldots, m}$ contains the roots of the Laguerre polynomial $L_{m}$ of $m$-th order and the weight function $\omega\left(u_h\right)$ is defined as  
\begin{equation}\label{weight_fun}
	\omega\left(u_h\right):= \frac{u_h}{\left(m + 1\right)^2 \left[L_{m+1}\left(u_h\right)\right]^2}.
\end{equation}
\end{proposition}

\begin{proof}
Let 
\begin{equation}\label{condexp_priceput}
	p\left(K, t_0, T\right) = e^{-r\left(T - t_0\right)}\mathbb{E}^{\mathbb{Q}}\left[\left(K-S_T\right)^+\left|\mathcal{F}_{t_0}\right.\right].
\end{equation}	
Denoting with $F\left(x\right)$ the conditional cumulative distribution function  of the log-price at maturity, the conditional expectation in \eqref{condexp_priceput} becomes
\begin{equation}\label{eq:IntegPut}
	 \mathbb{E}^{\mathbb{Q}}\left[\left(K-S_T\right)^+\left|\mathcal{F}_{t_0}\right.\right] = \int_{-\infty}^{+\infty}\left(e^{\kappa}-e^{u}\right)\mathbbm{1}_{\left\{u\leq k\right\}}\mbox{d}F\left(u\right), 
\end{equation}	
with $\kappa:=\ln\left(K\right)$ and $u:= \ln\left(S_T\right)$.

\noindent The right-hand side of \eqref{eq:IntegPut} can be written as 
\begin{eqnarray}
\int_{-\infty}^{+\infty}\left(e^{\kappa}-e^{u}\right)\mathbbm{1}_{\left\{u\leq k\right\}}\mbox{d}F\left(u\right)&=& \int_{-\infty}^{+\infty}\int_{u}^{\kappa}e^{t}\mbox{d}t\mbox{d}F\left(u\right)\nonumber\\
&=& \int_{-\infty}^{+\infty}e^{t}\left[\int_{-\infty}^{+\infty}\mathbbm{1}_{\left\{u\leq t \leq \kappa \right\}}\mbox{d}F\left(u\right)\right]\mbox{d}t\nonumber\\
&=& \int_{-\infty}^{+\infty}e^{t}\mathbbm{1}_{\left\{t\leq \kappa\right\}}\left[\int_{-\infty}^{t}\mbox{d}F\left(u\right)\right]\mbox{d}t\nonumber\\
&=& \int_{-\infty}^{\kappa}e^{t}F\left(t\right)\mbox{d}t.
\label{Intaaa}
\end{eqnarray}
Defining $y :=\kappa - t$, then the integral in \eqref{Intaaa} reads:
\begin{equation}
	\int_{-\infty}^{+\infty}\left(e^{\kappa}-e^{u}\right)\mathbbm{1}_{\left\{u\leq k\right\}}\mbox{d}F\left(u\right)=  e^{\kappa}\int_0^{+\infty}e^{-y}F\left(\kappa-y\right)\mbox{d}y.
	\label{cccc}
\end{equation}
The integral in \eqref{cccc} can be computed using the Gauss-Laguerre quadrature that leads to the following approximation for the European put price:
\begin{equation}\label{approx1aaa}
	p\left(K, t_0, T\right)  \approx e^{-r\left(T - t_0\right)}  K\sum_{j=1}^{m}F\left(\kappa-u_j\right)\omega\left(u_j\right),
\end{equation}
where the sequence $\left\{u_j\right\}_{j = 1,\ldots, m}$ contains the roots of the Laguerre polynomial $L_{m}$ of $m$-th order and $\omega\left(\cdot\right)$ is defined as in \eqref{weight_fun}. \newline
The cumulative distribution function of the log-price in \eqref{eq:STQ} has not a closed-form expression. However, it can be computed numerically by the inversion of the characteristic function \cite{Gil1951}. That is, 
\begin{equation}
	F\left(x\right)=\frac12  - \frac{1}{\pi}\int_0^{+\infty}\mathsf{Re}\left(\frac{e^{-iux}\phi(u)}{iu}\right)\mbox{d}u.
	\label{eq:Gil-Peleaz}
\end{equation}
Note that \eqref{eq:Gil-Peleaz} can be approximated by the Gauss-Laguerre quadrature as follows 
\begin{eqnarray}
	F\left(x\right) &=&\frac12  - \frac{1}{\pi}\int_0^{+\infty}\mathsf{Re}\left(\frac{e^{-iux} \phi(u)}{iu}\right)e^{u}e^{-u}\mbox{d}u \nonumber \\ 
	&\approx& \frac12 - \frac{1}{\pi}\sum_{k=1}^m \mathsf{Re}\left(\frac{e^{-i u_k x}\phi\left(u_k\right)}{i u_k}\right)e^{u_k} \omega\left(u_k\right).
	\label{approx2aaa}
\end{eqnarray}
Combining the approximations in \eqref{approx1aaa} and \eqref{approx2aaa}, we obtain the result in \eqref{put_numerical} that approximates the price of an European put option at time $t_0$.
\end{proof}
Note that the proposed formula for European put option \eqref{put_numerical} does not require the truncation of the range integration in the risk-neutral valuation formula as done in \cite{carr1999option} for the Carr-Madan formula and in \cite{fang2009novel} for the COS method. Furthermore, the Carr-Madan formula needs the identification of a damping factor which is necessary for the accuracy of the pricing formula, while the COS method introduces three approximation errors, namely i) the truncation of the range integration, ii) the truncation of Fourier cosine series, and iii) the approximation formula for the coefficients in the Fourier cosine series, as reported in \cite[p.~174]{oosterlee2019mathematical}. In our approach, the approximation error term in \eqref{put_numerical} arises from the Gauss-Laguerre quadrature and is controlled by the order $m$ of the Laguerre polynomial. European call option prices can obtained by means of put-call parity.


%% file: simulation_file.tex
In this section, we discuss a numerical method for simulating the underlying asset under the risk-neutral measure $\mathbb{Q}$ in which it is assumed that $\sigma_t$ is constant. Thus, the undelying asset price in \eqref{eq:solutionS_generalunderQ} writes:
\begin{equation}
S_{T}= S_{t_0}\exp\left[\left(r-\frac12 \sigma^2\right)\left(T-t_0\right)+\sigma\left(W_{T}-W_{t_0}\right) - \mathbb{E}^{\mathbb{Q}}\left[e^{\bar{J}^{\mathbb{Q}}}-1\right]\int_{t_0}^{T}\lambda_t \mbox{d}t + \left(Y^\mathbb{Q}_{T}-Y^\mathbb{Q}_{t_0}\right)\right].
\label{eq:SolutionSunderQ}
\end{equation}
The terms $\mathbb{E}^{\mathbb{Q}}\left[e^{\bar{J}^{\mathbb{Q}}}-1\right]\int_{t_0}^{T}\lambda_t \mbox{d}t$ and $\left(Y^\mathbb{Q}_{T}-Y^\mathbb{Q}_{t_0}\right)$ in \eqref{eq:SolutionSunderQ} require the simulation of the sequence of time arrivals and the counting process $N_t$.
In Section \ref{uCARMA}, we show that the intensity of a CARMA(p,q)-Hawkes process can be bounded by the intensity of a Hawkes process with an exponential kernel. This result enables us to construct an efficient thinning algorithm for simulating the time arrivals of the CARMA(p,q)-Hawkes process. Once the time arrivals are simulated and the final value of the counting process $N_T$ is obtained, we can compute the increment $Y^\mathbb{Q}_{T}-Y^\mathbb{Q}_{t_0}$ of the compound CARMA(p,q)-Hawkes process by generating $N_T - N_{t_0}$ jump sizes.
In Section~\ref{sim_compensator} we derive a recursive procedure for the evaluation of $\int_{t_0}^{T}\lambda_t \mbox{d}t$ once the time arrivals from $t_0$ to $T$ are observed.

\subsection{Simulation of  a CARMA(p,q)-Hawkes model through a thinning algorithm}\label{uCARMA}

We assume matrix $\mathbf{A}$ in \eqref{Acomp} to be diagonalizable, corresponding to the assumption that all eigenvalues of $\mathbf{A}$ are distinct. The eigenvectors of $\mathbf{A}$, i.e.,
$\left[1,\tilde{\lambda}_j, \tilde{\lambda}_j^2,\ldots,\tilde{\lambda}_J^{p-1}\right]^{\top}$ for $j=1,\ldots,p$, are used to define a $p\times p$ matrix $\mathbf{S}$. Specifically,
\begin{equation}
	\mathbf{S} :=\left[
	\begin{array}{ccc}
		1 & \ldots & 1\\
		\tilde{\lambda}_1 & \ldots & \tilde{\lambda}_p\\
		\tilde{\lambda}_1^2 & \ldots & \tilde{\lambda}_p^2\\
		\vdots & & \vdots\\
		\tilde{\lambda}_1^{p-1} & \ldots & \tilde{\lambda}_p^{p-1}\\
	\end{array}
	\right].
	\label{BoldS}
\end{equation}
It follows that $\mathbf{S}$ satisfies $\mathbf{S}^{-1}\mathbf{A}\mathbf{S}=\mathbf{\Lambda}$, where the diagonal matrix $\Lambda \in \mathbb{R}^{p \times p}$ is
\begin{equation}
	\mathbf{\Lambda} = \mathsf{diag}\left(\tilde{\lambda}_1,\ldots,\tilde{\lambda}_p\right).
	\label{BoldLambda}
\end{equation}
Quantities \eqref{BoldS} and \eqref{BoldLambda} play a crucial role in the construction of our thinning simulation algorithm. In Theorem \ref{thm1} we determine an upper bound for the intensity of a CARMA(p,q)-Hawkes process.
\begin{theorem}\label{thm1}
The intensity $\lambda_t$ of a CARMA(p,q)-Hawkes process is bounded from the following quantity
\begin{equation}
\bar{\lambda}_t := \mu + \sum_{T_{i}<t} \left\|\mathbf{b}^{\top}\mathbf{S}\right\|_2\left\|\mathbf{S}^{-1}\mathbf{e}\right\|_2 e^{\lambda\left(\mathbf{A}\right)\left(t-T_i\right)},
\label{ubint}
\end{equation}
where $T_i$ is the time arrival, $\left\|\cdot\right\|_{2}$ denotes the $\mathbb{L}^2$-norm, and $\lambda\left(\mathbf{A}\right)$ represents the largest real part eigenvalue of $\mathbf{A}$. 
\end{theorem}

\begin{proof}
The intensity of a CARMA(p,q)-Hawkes 
\begin{equation}\label{eq1_proof1}
\lambda_t = \mu +\mathbf{b}^\top\sum_{T_{i}<t} e^{\mathbf{A}\left(t-T_i\right)}\mathbf{e},
\end{equation}
can be rewritten using  the non-negativity condition for the CARMA(p,q)-Hawkes kernel \cite[Proposition 2]{MERCURI20241} as follows 
\begin{equation}
\lambda_t  =  \mu+ \sum_{T_{i}<t}\left|\mathbf{b}^\top e^{\mathbf{A}\left(t-T_i\right)}\mathbf{e}\right|.
\end{equation}
Using the fact that the companion matrix $\mathbf{A}$ is diagonizable, we get
\begin{equation}
\lambda_t  =  \mu+ \sum_{T_{i}<t}\left|\mathbf{b}^\top \mathbf{S} e^{\mathbf{\Lambda}\left(t-T_i\right)}\mathbf{S}^{-1}\mathbf{e}\right|.
\end{equation}
Using the matrix norm $\left\|\cdot\right\|_{M,2}$ induced by the Euclidean norm\footnote{
Let $A \in \mathbb{C}^{p\times p}$ and $\left\| \cdot \right\|_2 $ be the standard Euclidean norm in $\mathbb{C}^{p}$. The induced matrix norm $\left\| \cdot \right\|_{M,2}$ is defined as 
\begin{equation*}
\left\| A \right\|_{M,2}:= \sup_{x\neq 0}\frac{\left\|Ax\right\|_2}{\left\|x\right\|_2} = \sup_{\left\|y\right\|_2=1} \left\|Ay\right\|_2.
\end{equation*}
The matrix norm $\left\|\cdot\right\|_{M,2}$ satisfies the following two properties used in this paper: i) $\left\|Ax\right\|_2\leq \left\| A \right\|_{M,2} \left\| x \right\|_{2}$ where $x \in \mathbb{C}^p$ and ii) $\left\| e^{\mathbf{\Lambda} t}\right\|_{M,2}\leq e^{\lambda\left(\mathbf{\Lambda}\right)t}$, $t \geq 0$, where $\mathbf{\Lambda} \in \mathbb{C}^{p\times p}$ is a complex diagonal matrix and $\lambda\left(\mathbf{\Lambda}\right)$ is the real part of the largest diagonal entry. That is, $\lambda\left(\mathbf{\Lambda}\right): =\max_{i = 1, \ldots, p} \mathsf{Re}\left(\tilde{\lambda}_i\right)$ with $\tilde{\lambda}_{i}$ denoting the $i$-th diagonal entry. Refer to \cite{Lewis3800260304, meyer2023matrix} and reference therein for further details.}, we observe that
\begin{eqnarray*}
\lambda_t  &\leq& \mu+ \sum_{T_{i}<t} \left\|\mathbf{b}^{\top}\mathbf{S}\right\|_2\left\|e^{\mathbf{\Lambda}\left(t-T_i\right)}\mathbf{S}^{-1}\mathbf{e}\right\|_{2}\nonumber\\
&\leq& \mu+ \sum_{T_{i}<t} \left\|\mathbf{b}^{\top}\mathbf{S}\right\|_2\left\|\mathbf{S}^{-1}\mathbf{e}\right\|_{2}\left\|e^{\mathbf{\Lambda}\left(t-T_i\right)}\right\|_{M,2} \nonumber \\
&\leq& \mu + \sum_{T_{i}<t} \left\|\mathbf{b}^{\top}\mathbf{S}\right\|_2\left\|\mathbf{S}^{-1}\mathbf{e}\right\|_2 e^{\lambda\left(\mathbf{A}\right)\left(t-T_i\right)}.
\end{eqnarray*}
\end{proof}

\begin{remark}
The upper bound $\bar{\lambda}_t$ in \eqref{ubint} can be equivalently reformulated as
\begin{equation}
\bar{\lambda}_t = \mu + \sum_{T_{i}<t} \sqrt{\sum_{j=1}^{p}b\left(\tilde{\lambda}_j\right)^2}\sqrt{\sum_{j=1}^{p}\frac{1}{a^\prime\left(\tilde{\lambda}_j\right)^2}} e^{\lambda\left(\mathbf{A}\right)\left(t-T_i\right)},
\end{equation}
where the polynomials $a^\prime\left(\tilde{\lambda}_j\right)$ and $b\left(\tilde{\lambda}_j\right)$ are respectively defined as
\begin{equation}
a^\prime\left(\tilde{\lambda}_j\right):=p\tilde{\lambda}_j^{p-1}+a_{1}\left(p-1\right)\tilde{\lambda}_j^{p-2}+\ldots+a_{p-1}
\end{equation}
and
\begin{equation}
b\left(\tilde{\lambda}_j\right):=b_0+b_1\tilde{\lambda}_j+\ldots+b_{p-1}\tilde{\lambda}_j^{p-1}.
\end{equation}
\end{remark}
Inspired by the thinning algorithm proposed in \cite{Lewis3800260304} for inhomogeneous Poisson processes and adapted for the Hawkes process with an exponential kernel by Ogata  \citep[see][for further details]{Ogata1981}, we outline in Algorithm \ref{alg1} a novel simulation algorithm that through the result in Theorem \ref{thm1} overcomes the need for numerical solutions in the simulation of a CARMA(p,q)-Hawkes model \citep[see][supplementary material]{MERCURI20241}. 
\begin{algorithm}[h]
\scriptsize
\caption{Thinning algorithm for simulating a CARMA(p,q)-Hawkes model on $\left[0, \mathbb{T}\right]$}
\label{alg1}
\begin{algorithmic}[1]
\Statex \textbf{Input:} Orders $p$ and $q$; parameters $\{a_1, a_2, \ldots, a_p,b_0,b_1, \ldots, b_{q-1}, \mu\}$ and final time $\mathbb{T}$.
\Statex \textbf{Output:} Set of jump times $\mathsf{T}$.
\State \textbf{Set initial conditions:}
\Statex \hspace*{\algorithmicindent} Set $T_0=0$
\Statex \hspace*{\algorithmicindent} Set $\lambda_{T_0}=\mu$
\Statex \hspace*{\algorithmicindent} Set $t=0$
\Statex \hspace*{\algorithmicindent} Initialize the set of jump times $\mathsf{T}$ as $\left\{\emptyset\right\}$
\State \textbf{Generate the first jump time:}
\Statex \hspace*{\algorithmicindent} Simulate a uniform random number $u_1 \sim \mathcal{U}_{\left[0,1\right]}$
\Statex \hspace*{\algorithmicindent} Compute $T_1 = -\frac{\ln\left(u_1\right)}{\mu}$
\Statex \hspace*{\algorithmicindent} If $T_1 > \mathbb{T}$, terminate the algorithm (no jumps on the interval $\left[0,\mathbb{T}\right]$)
\Statex \hspace*{\algorithmicindent} Otherwise, add $T_1$ to $\mathsf{T}$ and set $t=T_1$ 
\State \textbf{while $t \leq \mathbb{T}$ do:}
\Statex \hspace*{\algorithmicindent} \textbf{Compute the intensity upper bound:}
\Statex \hspace*{\algorithmicindent} \hspace*{\algorithmicindent} Compute $\bar{\lambda}= \bar{\lambda}_t + \left\|\mathbf{b}^{\top}\mathbf{S}\right\|_2\left\|\mathbf{S}^{-1}\mathbf{e}\right\|_2$
\Statex \hspace*{\algorithmicindent} \textbf{Generate subsequent $t$:}
\Statex \hspace*{\algorithmicindent} \hspace*{\algorithmicindent} Generate a random number $u$ from a continuous uniform distribution on the interval $\left[0,1\right]$
\Statex \hspace*{\algorithmicindent} \hspace*{\algorithmicindent} Compute $\Delta T = -\frac{\ln\left(u\right)}{\bar{\lambda}}$
\Statex \hspace*{\algorithmicindent} \hspace*{\algorithmicindent} Set $t = t + \Delta T$
\Statex \hspace*{\algorithmicindent} \hspace*{\algorithmicindent} If $t > \mathbb{T}$, terminate the simulation algorithm and output $\mathsf{T}$
\Statex \hspace*{\algorithmicindent} \textbf{Decide if a jump occurs at time $t$:}
\Statex \hspace*{\algorithmicindent} \hspace*{\algorithmicindent} Generate $D$ from a continuous uniform distribution in the interval $\left[0,1\right]$
\Statex \hspace*{\algorithmicindent} \hspace*{\algorithmicindent} If $D\bar{\lambda} \leq \lambda_t$, add $t$ to $\mathsf{T}$
\Statex \textbf{end}
\State \textbf{return output $\mathsf{T}$}
\end{algorithmic}
\end{algorithm}

\subsection{Simulation of the  compensator}\label{sim_compensator}

To simulate the compensator term in \eqref{eq:SolutionSunderQ}, we compute the integral $\int_{t_0}^T\lambda_{t}\mbox{d}t$ given a specific sequence of time arrivals $\left\{T_i\right\}_{i \geq 1}$. Let $T_{k}$ be the last jump time before the final time $T>t_0$, we split the integral of the compensator as follows: 
\begin{equation}
\int_{t_0}^{T}\lambda_{t}\mbox{d}t=\int_{t_0}^{T_k}\lambda_t\mbox{d}t + \int_{T_k}^{T}\lambda_{t}\mbox{d}t.
\label{eq:int_tot}
\end{equation}
We consider the second component in \eqref{eq:int_tot} and we get:

\begin{equation*}
\int_{T_k}^{T}\lambda_t\mbox{d}t=\int_{T_k}^{T}\left[\mu + \mathbf{b}^{\top}X_{t}\right]\mbox{d}t.
\end{equation*}
Given a generic initial value $X_{t_0}$, the state process $\left\{X_t\right\}_{t\geq t_0}$ has the following form:
\begin{equation}
X_{t}= e^{\mathbf{A}\left(t-t_0\right)}X_{t_0}+\int_{\left[t_0, t\right)}e^{\mathbf{A}\left(t-u\right)}\mathbf{e}\mbox{d}N_{u}.
\label{eq:SolXtgivenXt0}
\end{equation}
Substituting \eqref{eq:SolXtgivenXt0} into the second integral in \eqref{eq:int_tot}, we get:
\begin{eqnarray}
\int_{T_k}^{T}\lambda_t\mbox{d}t&=& \int_{T_k}^{T}\left[\mu+\mathbf{b}^{\top}\left(e^{\mathbf{A}\left(t-t_0\right)}X_{t_0}+\int_{\left[t_0,t\right)}e^{\mathbf{A}\left(t-u\right)}\mathbf{e}\mbox{d}N_u\right)\right]\mbox{d}t\nonumber\\
&=&\mu\left(T-T_k\right)+\mathbf{b}^{\top}\left[\mathbf{A}^{-1}\left(e^{\mathbf{A}\left(T-t_0\right)}-e^{\mathbf{A}\left(T_k-T_0\right)}\right)\right]X_{t_0}+ \mathbf{b}^{\top}\left[S\left(k\right)-e^{\mathbf{A}\left(T_k-T_{k_0}\right)}S\left(k_0\right)\right]\mathbf{A}^{-1}\left(e^{\mathbf{A}\left(T-T_k\right)}-\mathbf{I}\right)\mathbf{e}\nonumber\\
\label{eq2}
\end{eqnarray}
where $k$ and $k_0$ are the numbers of jumps in the interval $\left[0,T\right]$ and $\left[0,t_0\right]$ respectively. The term $S\left(k\right)$ is defined as
\begin{equation}\label{eq:S_k}
S\left(k\right):= \sum_{i=1}^{k}e^{\mathbf{A}\left(T_k-T_i\right)}
\end{equation}
and \cite{MERCURI20241} show that $S\left(k\right)$ satisfies the following recursive equation:
\begin{equation*}
S\left(j\right) = e^{\mathbf{A}\left(T_{j}-T_{j-1}\right)}S\left(j-1\right)+\mathbf{I}
\end{equation*}
with the initial condition $S\left(1\right)=\mathbf{I}$ and $\mathbf{I}$ is a $p\times p$ identity matrix .

\noindent The first integral in the rhs of \eqref{eq:int_tot} reads:
\begin{equation}
\int_{t_0}^{T_k}\lambda_t\mbox{d}t=\int_{t_0}^{T_k}\left[\mu+\mathbf{b}^{\top}\left(e^{\mathbf{A}\left(t-t_0\right)}X_{t_0}+\int_{t_0}^t e^{\mathbf{A}\left(t-u\right)}\mathbf{e}\mbox{d}N_u\right)\right]\mbox{d}t+\mathbf{b}^{\top}\int_{t_0}^{T_k}\left[\int_{t_0}^t e^{\mathbf{A}\left(t-u\right)}\mathbf{e}\mbox{d}N_u\right]\mbox{d}t,
\label{eq:firstInt}
\end{equation}
and the component $\mathbf{b}^{\top}\int_{t_0}^{T_k}\left[\int_{t_0}^t e^{\mathbf{A}\left(t-u\right)}\mathbf{e}\mbox{d}N_u\right]\mbox{d}t$ can be written as
\begin{equation*}
\mathbf{b}^{\top}\int_{t_0}^{T_k}\left[\int_{t_0}^t e^{\mathbf{A}\left(t-u\right)}\mathbf{e}\mbox{d}N_u\right]\mbox{d}t=\mathbf{b}^{\top}\int_{t_0}^{T_k}\left[\int_{t_0}^{+\infty}\mathbbm{1}_{\left\{u<t\leq T_k\right\}}e^{\mathbf{A}\left(t-u\right)}\mathbf{e}\mbox{d}N_u\right]\mbox{d}t.
\end{equation*}
Inverting the order of integration, we obtain:
\begin{equation}
\mathbf{b}^{\top}\int_{t_0}^{T_k}\left[\int_{t_0}^{+\infty}\mathbbm{1}_{\left\{u<t\leq T_k\right\}}e^{\mathbf{A}\left(t-u\right)}\mathbf{e}\mbox{d}N_u\right]\mbox{d}t=\mathbf{b}^{\top} \mathbf{A}^{-1}\left[S\left(k\right)-e^{\mathbf{A}\left(T_k-T_{k_0}\right)}S\left(k_0\right)\right]\mathbf{e}-\mathbf{b}^{\top} \mathbf{A}^{-1}\mathbf{e}\left(k-k_0\right).
\label{eq:nn}
\end{equation}
Substituting \eqref{eq:nn} in \eqref{eq:firstInt}, we have:
\begin{eqnarray}
\int_{t_0}^{T_k}\lambda_{t}\mbox{d}t&=&\mu\left(T_k-t_0\right)+\mathbf{b}^{\top}\mathbf{A}^{-1}\left[e^{\mathbf{A}\left(T_k-t_0\right)}-\mathbf{I}\right]X_{t_0}\nonumber\\
&+&\mathbf{b}^{\top} \mathbf{A}^{-1}\left[S\left(k\right)-e^{\mathbf{A}\left(T_k-T_{k_0}\right)}S\left(k_0\right)\right]\mathbf{e}-\mathbf{b}^{\top} \mathbf{A}^{-1}\mathbf{e}\left(k-k_0\right).\nonumber\\
\label{eq1}
\end{eqnarray}
Using the result in \eqref{eq1} and \eqref{eq2}, the quantity in \eqref{eq:int_tot} becomes:
\begin{equation}
\int_{t_0}^{T}\lambda_t\mbox{d}t = \mu \left(T-t_0\right)+\mathbf{b}^{\top}A^{-1}\left(e^{\mathbf{A}\left(T-t_0\right)}-\mathbf{I}\right)X_{t_0}-\mathbf{b}^{\top}A^{-1}\mathbf{e}\left(k-k_0\right)+\mathbf{b}^{\top}A^{-1}\left[S\left(k\right)e^{\mathbf{A}\left(T-T_k\right)}-e^{\mathbf{A}\left(T_k-T_{k_0}\right)}S\left(k_0\right)\right]\mathbf{e}.
\end{equation}
\begin{remark}
Choosing $t_0=0$ and $X_{t_0}=\mathbf{0}$, the integrated intensity simplifies as follows:
\begin{equation}
\int_{t_0}^{T}\lambda_t\mbox{d}t=\mu T-\mathbf{b}^{\top}A^{-1}\mathbf{e}k+\mathbf{b}^{\top}A^{-1}S\left(k\right)e^{\mathbf{A}\left(T-T_k\right)}\mathbf{e}.
\end{equation}
\end{remark}

%% file: numerical_results.tex
In this section we present several numerical examples in order to verify the reliability, accuracy, and efficiency of the methodology developed in Section~\ref{GaussLaguerre}. The analysis is conducted by pricing European call options for three different models\footnote[5]{For the sake of clarity, and throughout the paper, we abbreviate the compound version of these models as Hawkes, CARMA(2,1)-Hawkes and CARMA(3,1)-Hawkes.}, namely Hawkes, CARMA(2,1)-Hawkes, and CARMA(3,1)-Hawkes, in which the jump size is normally distributed with mean $\mu_J$ and variance $\sigma^2_J$; see Table~\ref{tab: numerical_params} for the parameter setting. 

\begin{table}[h]
	\centering
	\caption{Parameter setting for Hawkes, CARMA(2,1)-Hawkes, and  CARMA(3,1)-Hawkes.}
	\label{tab: numerical_params}
	\begin{tabular}{cccc}
		\toprule
		& Hawkes & CARMA(2,1)-Hawkes & CARMA(3,1)-Hawkes  \\ \midrule
		$\mu$      & 3.00   & 3.00                & 3.00                    \\
		$b_0$      & 1.00   & 1.00                & 0.20                  \\
		$b_1$      & -      & 0.30                & 0.30                  \\
		$a_1$      & 3.00   & 3.00                & 1.30                  \\
		$a_2$      & -      & 2.00                & $0.34 + \pi^2/4$     \\
		$a_3$      & -      & -                   & $0.025 +0.025 \pi^2$ \\
		$\mu_J$    & 0.00   & 0.00                & 0.00                    \\
		$\sigma_J$ & 0.45   & 0.45                & 0.45                 \\
		$\sigma$ & 0.20   & 0.20                & 0.20                  \\
		$S_0$      & 100    & 100                 & 100                  \\
		$r$        & 0.05   & 0.05                & 0.05                 \\ \bottomrule
	\end{tabular}
\end{table}

The numerical experiments are run in R software (version 4.2.2) with an Intel Core i7-8565U CPU @ 1.80GHz and a 16GB RAM memory.

\subsection{Prices of European call options}

The European call option prices are computed by means of Proposition~\ref{theo_gl} (obtained through put-call parity) using the characteristic function of the log-price under $\mathbb{Q}$ measure proposed in \eqref{cf_general} and its time-dependent coefficients $u_{0, T}\left(t_0\right)$ and $u_{2, T}\left(t_0\right)$, satisfying the ODE's system, are solved numerically through the Euler method with a number of discretization points equal to $\bar{n} = 2\cdot 10^3$. As concerns the formula in Proposition~\ref{theo_gl}, we fix the order $m$ of Laguerre polynomials to $m = 4.5 \cdot 10^2$. We consider six levels of strike prices, $K = \{70, 80, 90, 100, 110, 120\}$, and four maturities $T = \{0.25, 0.5, 1, 3.5\}$ in years.

The obtained outcomes are then compared with the pricing results of Monte Carlo simulation, based on the exact solution \eqref{eq:SolutionSunderQ}, using $M = 10^6$ simulations and control variates technique for reducing the variance obtained from the simulation of the underlying asset price.\footnote{We use $\mathbb{E}\left[S_T|\mathcal{F}_{t_0}\right]$ as a control variable. This quantity can be determined directly from the characteristic function in \eqref{cf_general}.}
For each maturity, Figure~\ref{fig: jump_hists} shows a snapshot of the relative frequency distribution of the number of simulated jumps occurred at time $T$ with the specific aim of testing the prices obtained by our approach in different scenarios. 

Table~\ref{tab: call_prices} reports the prices of European call options computed through Proposition~\ref{theo_gl} ($C_{theo}$) and Monte Carlo simulation ($C_{MC}$) with its $95\%$ confidence interval (denoted by lower bound $LB$ and upper bound $UB$) for Hawkes, CARMA(2,1)-Hawkes, and CARMA(3,1)-Hawkes models. We observe that all the obtained pricing results are close to those of Monte Carlo Simulation and fall into the $95\%$ confidence interval, validating thus the general model and the methodology. In Figure~\ref{fig: iv} we display the Black-Scholes implied volatilities (herafter also IV) for the three models under investigation and for each maturity in a range of strike prices $K \in [70, 120]$ with a step size of two. The Hawkes, CARMA(2,1)-Hawkes, and CARMA(3,1)-Hawkes models exhibit the same behaviour, but the Hawkes model shows higher implied volatilities $T = \{0.25, 0.5, 1\}$ with respect to the other two models with the parameter setting of Table~\ref{tab: numerical_params}.

\begin{sidewaysfigure}
	\centering
	\caption{Relative frequency distribution of $N_T$ for $T = \{0.25, 0.5, 1, 3.5\}$. Parameters are summarised in Table~\ref{tab: numerical_params}.}
	\begin{subfigure}[b]{0.495\textwidth}
		\centering
		\includegraphics[width=\textwidth]{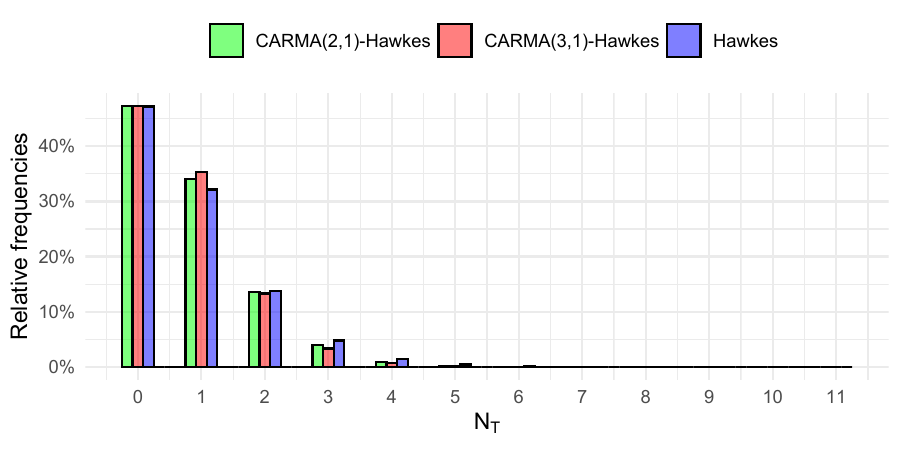} 
		\caption{Time to maturity $T = 0.25$} 
		\label{fig7:a} 
	\end{subfigure}
	\vspace{4ex}
	\begin{subfigure}[b]{0.495\textwidth}
		\centering
		\includegraphics[width=\textwidth]{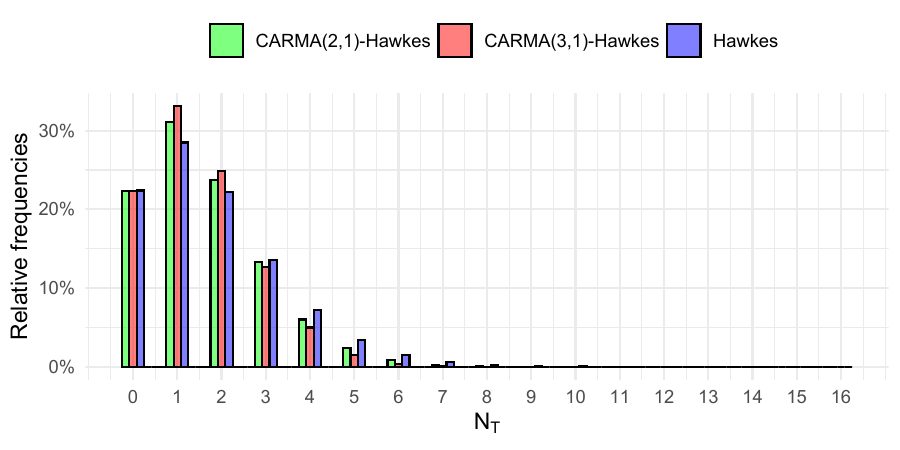} 
		\caption{Time to maturity $T = 0.5$} 
		\label{fig7:b} 
	\end{subfigure} 
	\vspace{4ex}
	\begin{subfigure}[b]{0.495\textwidth}
		\centering
		\includegraphics[width=\textwidth]{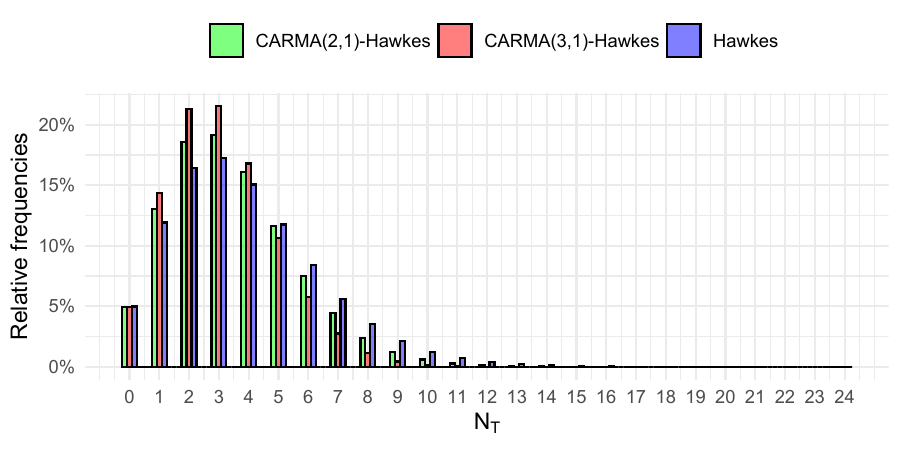} 
		\caption{Time to maturity $T = 1$} 
	\end{subfigure}
	\vspace{4ex}
	\begin{subfigure}[b]{0.495\textwidth}
		\centering
		\includegraphics[width=\textwidth]{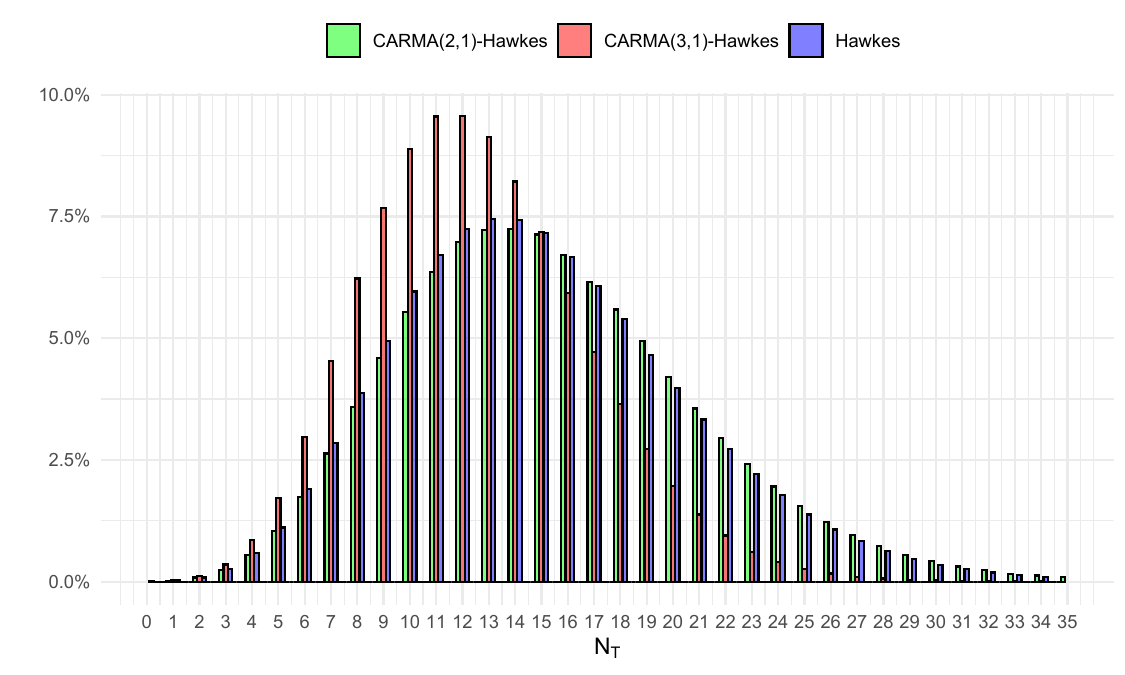} 
		\caption{Time to maturity $T = 3.5$} 
		\label{fig7: c} 
	\end{subfigure} 
	\label{fig: jump_hists} 
\end{sidewaysfigure}

\begin{sidewaystable}
	\caption{Prices of European call options for $T = \{0.25, 0.5, 1, 3.5\}$. Parameters are summarised in Table~\ref{tab: numerical_params}.}
	\begin{subtable}[t]{0.5\linewidth}
		\centering
		\begin{tabular}{ccccccc}
			\toprule
			Strikes    & 70      & 80      & 90      & 100     & 110     & 120     \\ \midrule
			\multicolumn{7}{c}{\begin{tabular}[c]{@{}c@{}}Hawkes\end{tabular}} \\ \hline
			$C_{theo}$ & 33.9953 & 26.0604 & 19.4060 & 14.9706 & 12.3186 & 10.4774 \\
			$C_{MC}$   & 33.9908 & 26.0584 & 19.4054 & 14.9722 & 12.3213 & 10.4800 \\
			LB         & 33.9747 & 26.0372 & 19.3792 & 14.9425 & 12.2895 & 10.4468 \\
			UB         & 34.0069 & 26.0797 & 19.4315 & 15.0019 & 12.3531 & 10.5132 \\ \hline
			\multicolumn{7}{c}{\begin{tabular}[c]{@{}c@{}}CARMA(2,1)-Hawkes\end{tabular}} \\ \hline
			$C_{theo}$ & 33.8009 & 25.8182 & 19.1223 & 14.6535 & 11.9762 & 10.1174 \\
			$C_{MC}$   & 33.8051 & 25.8207 & 19.1244 & 14.6561 & 11.9771 & 10.1154 \\
			LB         & 33.7897 & 25.8003 & 19.0992 & 14.6275 & 11.9466 & 10.0837 \\
			UB         & 33.8205 & 25.8411 & 19.1496 & 14.6847 & 12.0076 & 10.1471 \\ \hline
			\multicolumn{7}{c}{\begin{tabular}[c]{@{}c@{}}CARMA(3,1)-Hawkes\end{tabular}}
			\\ \hline
			$C_{theo}$ & 33.6955 & 25.6859 & 18.9667 & 14.4790 & 11.7875 & 9.9190  \\
			$C_{MC}$   & 33.6928 & 25.6829 & 18.9607 & 14.4727 & 11.7813 & 9.9132  \\
			LB         & 33.6779 & 25.6630 & 18.9361 & 14.4448 & 11.7515 & 9.8823   \\
			UB         & 33.7077 & 25.7028 & 18.9853 & 14.5007 & 11.811  & 9.9441  \\ \bottomrule
		\end{tabular}
		\caption{Time to maturity $T = 0.25$}
	\end{subtable}
	\hspace{\fill}
	\begin{subtable}[t]{0.5\linewidth}
		\centering
		\begin{tabular}{ccccccc}
			\toprule
			Strikes    & 70      & 80      & 90      & 100     & 110     & 120     \\ \midrule
			\multicolumn{7}{c}{\begin{tabular}[c]{@{}c@{}}Hawkes\end{tabular}} \\ \hline
			$C_{theo}$ & 38.7475 & 32.6015 & 27.6842 & 23.9358 & 21.0241 & 18.6542 \\
			$C_{MC}$   & 38.7529 & 32.6075 & 27.6912 & 23.9457 & 21.0369 & 18.6678 \\
			LB         & 38.7281 & 32.5768 & 27.6551 & 23.9050 & 20.9921 & 18.6195 \\
			UB         & 38.7777 & 32.6382 & 27.7273 & 23.9865 & 21.0816 & 18.7161 \\ \hline
			\multicolumn{7}{c}{\begin{tabular}[c]{@{}c@{}}CARMA(2,1)-Hawkes\end{tabular}} \\ \hline
			$C_{theo}$ & 38.2058 & 31.9557 & 26.9502 & 23.1308 & 20.1645 & 17.7547 \\
			$C_{MC}$   & 38.2066 & 31.9567 & 26.9493 & 23.1281 & 20.1633 & 17.7548 \\
			LB         & 38.1837 & 31.9282 & 26.9159 & 23.0907 & 20.1226 & 17.7112 \\
			UB         & 38.2295 & 31.9851 & 26.9827 & 23.1655 & 20.2040 & 17.7983 \\ \hline
			\multicolumn{7}{c}{\begin{tabular}[c]{@{}c@{}}CARMA(3,1)-Hawkes\end{tabular}}
			\\ \hline
			$C_{theo}$ & 37.8390 & 31.5147 & 26.4465 & 22.5768 & 19.5722 & 17.1349  \\
			$C_{MC}$   & 37.8456 & 31.5260 & 26.4611 & 22.5932 & 19.5911 & 17.1550 \\
			LB         & 37.8236 & 31.4986 & 26.4288 & 22.5570 & 19.5516 & 17.1129   \\
			UB         & 37.8675 & 31.5535 & 26.4934 & 22.6295 & 19.6305 & 17.1972  \\ \bottomrule
		\end{tabular}
		\caption{Time to maturity $T = 0.5$}
	\end{subtable}
	\hspace{\fill}
	\begin{subtable}[t]{0.5\linewidth}
		\centering
		\begin{tabular}{ccccccc}
			\toprule
			Strikes    & 70      & 80      & 90      & 100     & 110     & 120     \\ \midrule
			\multicolumn{7}{c}{\begin{tabular}[c]{@{}c@{}}Hawkes\end{tabular}} \\ \hline
			$C_{theo}$ & 47.7542 & 43.3566 & 39.6128 & 36.4029 & 33.6213 & 31.1875 \\
			$C_{MC}$   & 47.7468 & 43.3468 & 39.6027 & 36.3924 & 33.6091 & 31.1739 \\
			LB         & 47.7123 & 43.3061 & 39.5562 & 36.3406 & 33.5523 & 31.1125 \\
			UB         & 47.7813 & 43.3875 & 39.6491 & 36.4442 & 33.6659 & 31.2354 \\ \hline
			\multicolumn{7}{c}{\begin{tabular}[c]{@{}c@{}}CARMA(2,1)-Hawkes\end{tabular}} \\ \hline
			$C_{theo}$ & 46.7325 & 42.1923 & 38.3279 & 35.0179 & 32.1546 & 29.6551 \\
			$C_{MC}$   & 46.7199 & 42.1759 & 38.3058 & 34.9908 & 32.1241 & 29.6225 \\
			LB         & 46.6874 & 42.1376 & 38.2620 & 34.9420 & 32.0707 & 29.5648 \\
			UB         & 46.7523 & 42.2142 & 38.3496 & 35.0396 & 32.1775 & 29.6801 \\ \hline
			\multicolumn{7}{c}{\begin{tabular}[c]{@{}c@{}}CARMA(3,1)-Hawkes\end{tabular}}
			\\ \hline
			$C_{theo}$ & 45.6790 & 40.9844 & 36.9902 & 33.5736 & 30.6247 & 28.0579  \\
			$C_{MC}$   & 45.6878 & 40.9948 & 37.0013 & 33.5842 & 30.6321 & 28.0628 \\
			LB         & 45.6570 & 40.9582 & 36.9593 & 33.5372 & 30.5806 & 28.0072  \\
			UB         & 45.7186 & 41.0314 & 37.0433 & 33.6311 & 30.6835 & 28.1184  \\ \bottomrule
		\end{tabular}
		\caption{Time to maturity $T = 1$}
	\end{subtable}
	\hspace{\fill}
	\begin{subtable}[t]{0.5\linewidth}
		\centering
		\begin{tabular}{ccccccc}
			\toprule
			Strikes    & 70      & 80      & 90      & 100     & 110     & 120     \\ \midrule
			\multicolumn{7}{c}{\begin{tabular}[c]{@{}c@{}}Hawkes\end{tabular}} \\ \hline
			$C_{theo}$ & 72.8056 & 70.7513 & 68.8858 & 67.1788 & 65.6069 & 64.1513 \\
			$C_{MC}$   & 72.8192 & 70.7671 & 68.9032 & 67.1973 & 65.6265 & 64.1723 \\
			LB         & 72.7754 & 70.7176 & 68.8483 & 67.1372 & 65.5615 & 64.1024 \\
			UB         & 72.8631 & 70.8166 & 68.9581 & 67.2574 & 65.6916 & 64.2422 \\ \hline
			\multicolumn{7}{c}{\begin{tabular}[c]{@{}c@{}}CARMA(2,1)-Hawkes\end{tabular}} \\ \hline
			$C_{theo}$ & 73.1905 & 71.1693 & 69.3331 & 67.6524 & 66.1040 & 64.6698 \\
			$C_{MC}$   & 73.1747 & 71.1513 & 69.3120 & 67.6277 & 66.0761 & 64.6400 \\
			LB         & 73.1305 & 71.1014 & 69.2566 & 67.5671 & 66.0104 & 64.5694 \\
			UB         & 73.2190 & 71.2013 & 69.3674 & 67.6883 & 66.1418 & 64.7106 \\ \hline
			\multicolumn{7}{c}{\begin{tabular}[c]{@{}c@{}}CARMA(3,1)-Hawkes\end{tabular}}
			\\ \hline
			$C_{theo}$ & 69.8727 & 67.5474 & 65.4388 & 63.5128 & 61.7427 & 60.1071   \\
			$C_{MC}$   & 69.8817 & 67.5581 & 65.4497 & 63.5228 & 61.7521 & 60.1170 \\
			LB         & 69.8381 & 67.5086 & 65.3946 & 63.4622 & 61.6863 & 60.0462  \\
			UB         & 69.9253 & 67.6075 & 65.5048 & 63.5833 & 61.8179 & 60.1879  \\ \bottomrule
		\end{tabular}
		\caption{Time to maturity $T = 3.5$}
	\end{subtable}	
	\label{tab: call_prices}
\end{sidewaystable}

\begin{sidewaysfigure}
	\centering
	\caption{Black-Scholes implied volatilities for $T = \{0.25, 0.5, 1, 3.5\}$. Parameters are summarised in Table~\ref{tab: numerical_params}.}
	\begin{subfigure}[b]{0.495\textwidth}
		\centering
		\includegraphics[width=\textwidth]{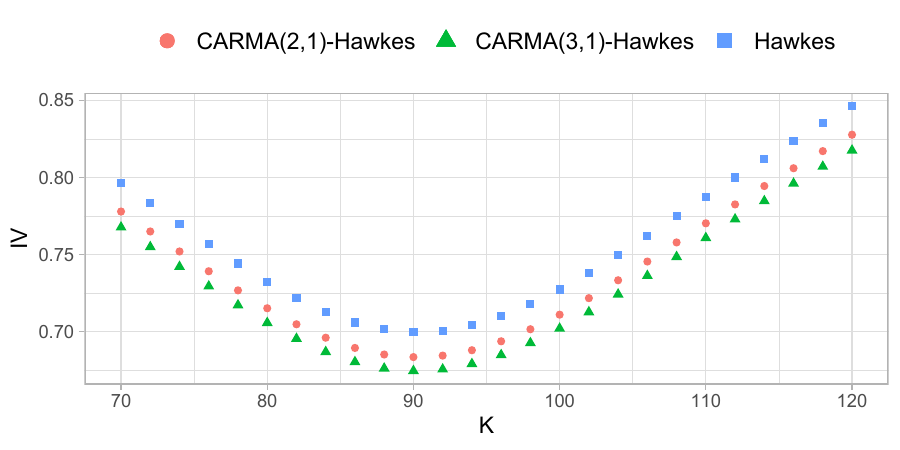} 
		\caption{Time to maturity $T = 0.25$} 
		\label{fig_iv: iv_025} 
	\end{subfigure}
	\vspace{4ex}
	\begin{subfigure}[b]{0.495\textwidth}
		\centering
		\includegraphics[width=\textwidth]{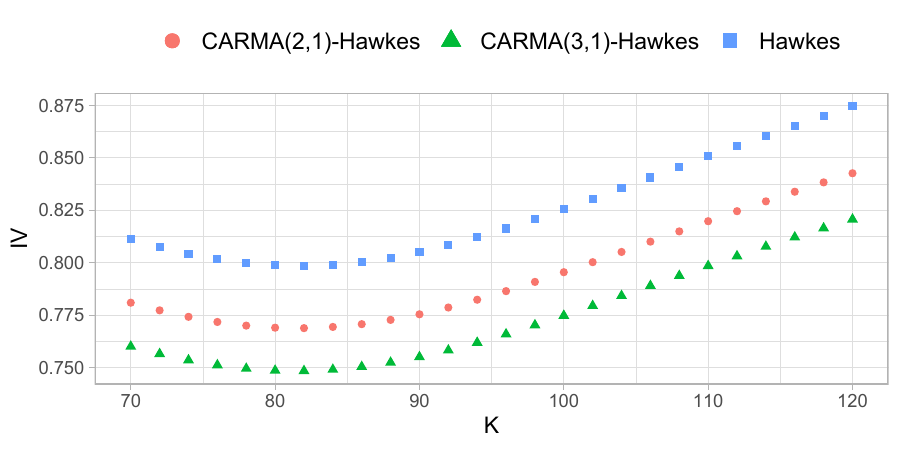} 
		\caption{Time to maturity $T = 0.5$} 
		\label{fig_iv: iv_05} 
	\end{subfigure} 
	\vspace{4ex}
	\begin{subfigure}[b]{0.495\textwidth}
		\centering
		\includegraphics[width=\textwidth]{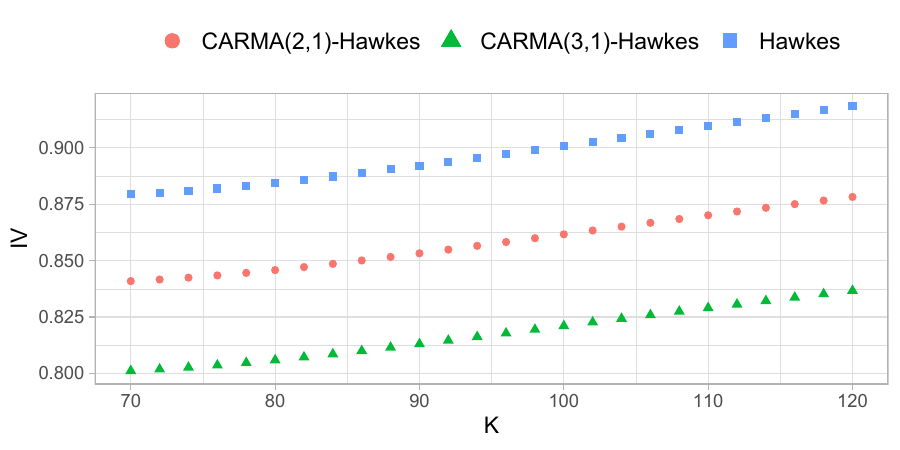} 
		\caption{Time to maturity $T = 1$} 
		\label{fig_iv: iv_1}
	\end{subfigure}
	\vspace{4ex}
	\begin{subfigure}[b]{0.495\textwidth}
		\centering
		\includegraphics[width=\textwidth]{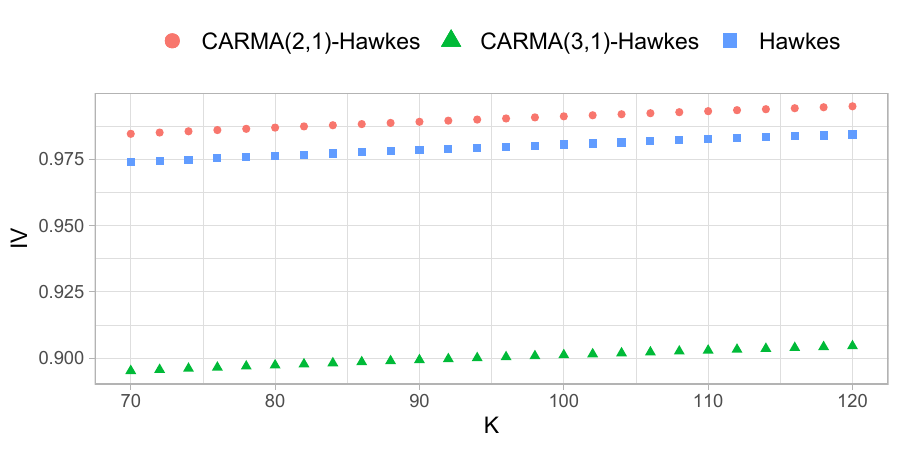} 
		\caption{Time to maturity $T = 3.5$} 
		\label{fig_iv: iv_35} 
	\end{subfigure} 
	\label{fig: iv} 
\end{sidewaysfigure}

\subsection{Sensitivity analysis}

The section presents a sensitivity analysis for the common parameters of the Hawkes, CARMA(2,1)-Hawkes, and CARMA(3,1)-Hawkes models that influence the European call option price. The comparison is reported in terms of Black-Scholes implied volatility for a maturity $T = 1$ and with parameters as stated in Table~\ref{tab: numerical_params} unless otherwise indicated. In examining the baseline intensity parameter $\mu$, the jump size parameters $(\mu_J, \sigma_J)$, and the volatility $\sigma$, the investigation is conducted as function of these parameters for at-the-money European call options (see Figure~\ref{fig: iv_sens}). In contrast, the sensitivity analysis of the implied volatility for the autoregressive and moving average parameters is undertaken as a function of the strike price (i.e., $K \in \left[40, 180\right]$ with step size equal to $10$) for varying levels of the aforementioned parameters (see Figure~\ref{fig: carma_sens}).

Figure~\ref{fig_iv_sens: mu} illustrates the behaviour of the implied volatility for the Hawkes, CARMA(2,1)-Hawkes, and CARMA(3,1)-Hawkes models when all parameters are held constant except for the baseline intensity parameter $\mu$, which varies within the interval $\left[0.3, 5.1\right]$. In accordance with expectations, the implied volatility increases with the parameter $\mu$, as this provides the base level for the conditional intensity at a given time $t$ (see \eqref{Def:Lambda}). It can be observed that for low values of $\mu$, the implied volatilities of the three models provide similar values. However, as the parameter $\mu$ increases, they tend to diverge. 

In Figure~\ref{fig_iv_sens: mu_j} we plot the values of IV for the Hawkes, CARMA(2,1)-Hawkes, and CARMA(3,1)-Hawkes models against the jump mean parameter $\mu_J$ within the specified range of $-0.5$ to $0.5$, observing a non-monotonic behaviour. For $\mu_J \leq 0$, the three models exhibit a small distance between themselves; however, as $\mu_J$ increases beyond zero, such distance tends to increase. 

Figure~\ref{fig_iv_sens: sigma_j} exhibits the influence of the jump standard deviation parameter $\sigma_J$ within the interval $\left[0, 1\right]$ on the implied volatility. In the case of $\sigma_J = 0$ (no contribution from the jump component of the jump-diffusion model), the implied volatility naturally is identical across the three models and corresponds to the volatility itself, i.e., $\sigma = 0.2$. It can be observed that for low values of $\sigma_J$, the Hawkes, CARMA(2,1)-Hawkes, and CARMA(3,1)-Hawkes models show a minimal discrepancy between themselves. Conversely, for high values of $\sigma_J$, the discrepancy tends to increase.

In Figure~\ref{fig_iv_sens: sigma_d}, for sake of completeness, we analyse the behaviour of implied volatility for the three models as a function of the parameter $\sigma$. As this parameter increases, naturally the implied volatility likewise rises.  However, it can be observed that as the value of $\sigma$ increases, the discrepancy between the three models in terms of IV diminishes, which could be attributed to a reduction in the contribution of the jump component to the overall jump-diffusion process.

The final set of parameters under investigation are the common autoregressive and moving average parameters shared by the Hawkes, CARMA(2,1)-Hawkes, and CARMA(3,1)-Hawkes models, i.e., $a_1$ and $b_0$. The objective is to determine the extent to which these factors influence the implied volatility. As illustrated in Figure~\ref{fig: carma_sens}, an increase in the parameter $a_1$ for all three models results in a downward shift in the volatility curves. The magnitude of this downward shift is most pronounced in the case of the Hawkes (Figure~\ref{fig: carma_sens_hawkes}) model, whereas it is less for the CARMA(2,1)-Hawkes model (Figure~\ref{fig: carma_carma_21}) and negligible for the CARMA(3,2)-Hawkes model (Figure~\ref{fig: carma_carma_32}). The opposite is observed with regard to the moving average parameter $b_0$ (see  Figure~\ref{fig: carma_sens}). Indeed, an increase in  $b_0$ is accompanied by an upward movement in the volatility curves.


\begin{sidewaysfigure}
	\centering
	\caption{Sensitivity analysis of the Black-Scholes implied volatility, evaluated using European call options that are at-the-money, as a function of the baseline intensity parameter $\mu$, the jump size parameters $(\mu_J, \sigma_J)$, and the volatility $\sigma$. The maturity is $T =1$ and the other parameters are summarised in Table~\ref{tab: numerical_params}.}
	\begin{subfigure}[b]{0.495\textwidth}
		\centering
		\includegraphics[width=\textwidth]{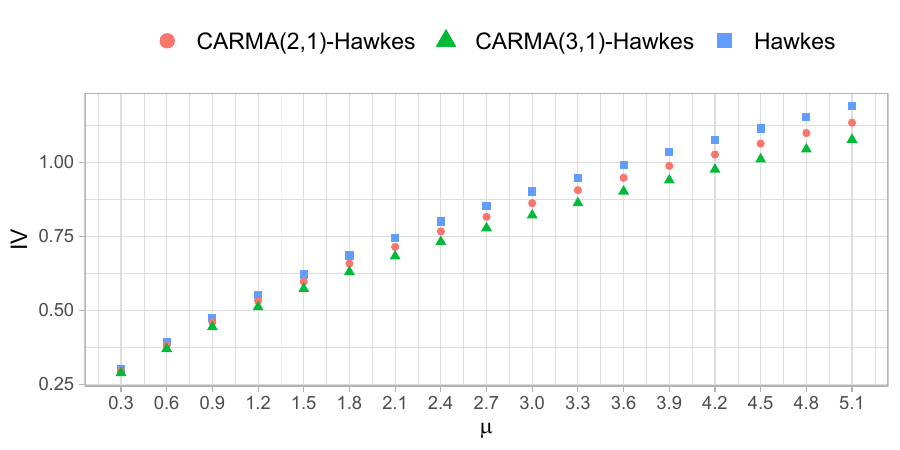} 
		\caption{baseline intensity parameter $\mu$} 
		\label{fig_iv_sens: mu} 
	\end{subfigure}
	\vspace{4ex}
	\begin{subfigure}[b]{0.495\textwidth}
		\centering
		\includegraphics[width=\textwidth]{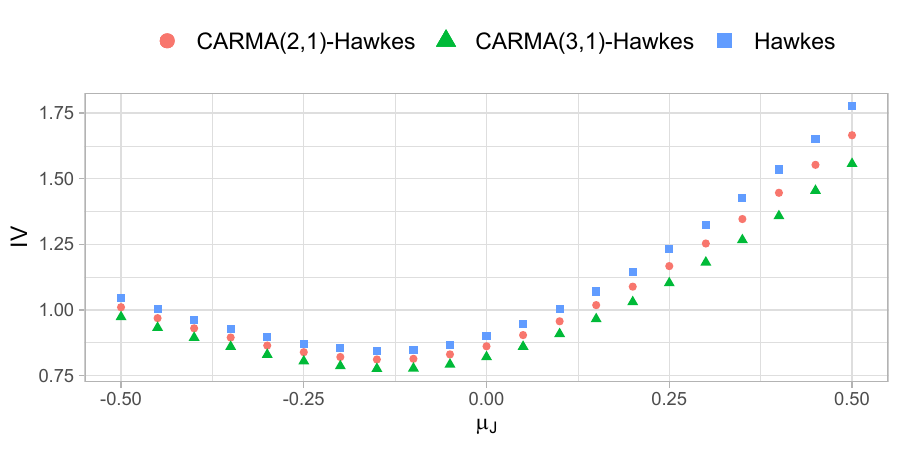} 
		\caption{jump size parameter $\mu_J$} 
		\label{fig_iv_sens: mu_j} 
	\end{subfigure} 
	\vspace{4ex}
	\begin{subfigure}[b]{0.495\textwidth}
		\centering
		\includegraphics[width=\textwidth]{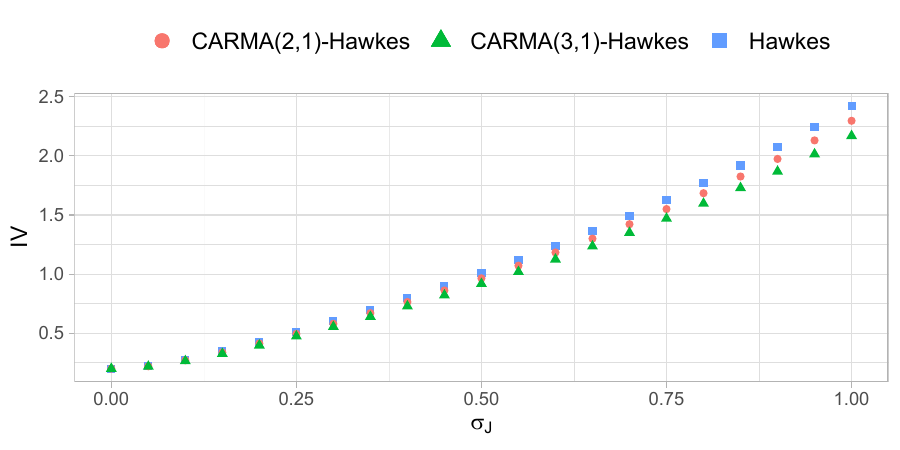} 
		\caption{jump size parameter $\sigma_J$} 
		\label{fig_iv_sens: sigma_j}
	\end{subfigure}
	\vspace{4ex}
	\begin{subfigure}[b]{0.495\textwidth}
		\centering
		\includegraphics[width=\textwidth]{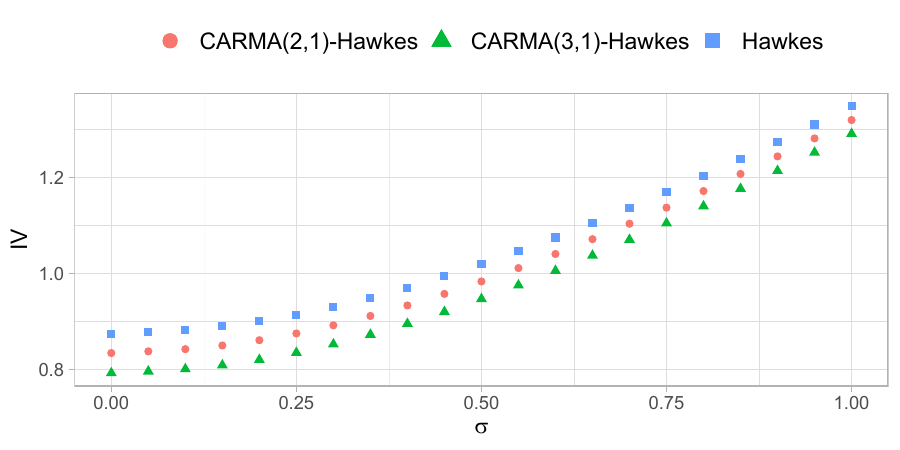} 
		\caption{volatility parameter $\sigma$} 
		\label{fig_iv_sens: sigma_d} 
	\end{subfigure} 
	\label{fig: iv_sens} 
\end{sidewaysfigure}


\begin{figure}
	\centering
	\caption{Sensitivity analysis of the Black-Scholes implied volatility evaluated from European call options as a function of the strike price $K$ for varying levels of the autoregressive and moving average parameters that are in common. The maturity is $T =1$ and the other parameters are summarised in Table~\ref{tab: numerical_params}.}
	\begin{subfigure}{0.81\textwidth}
		\includegraphics[width=\textwidth]{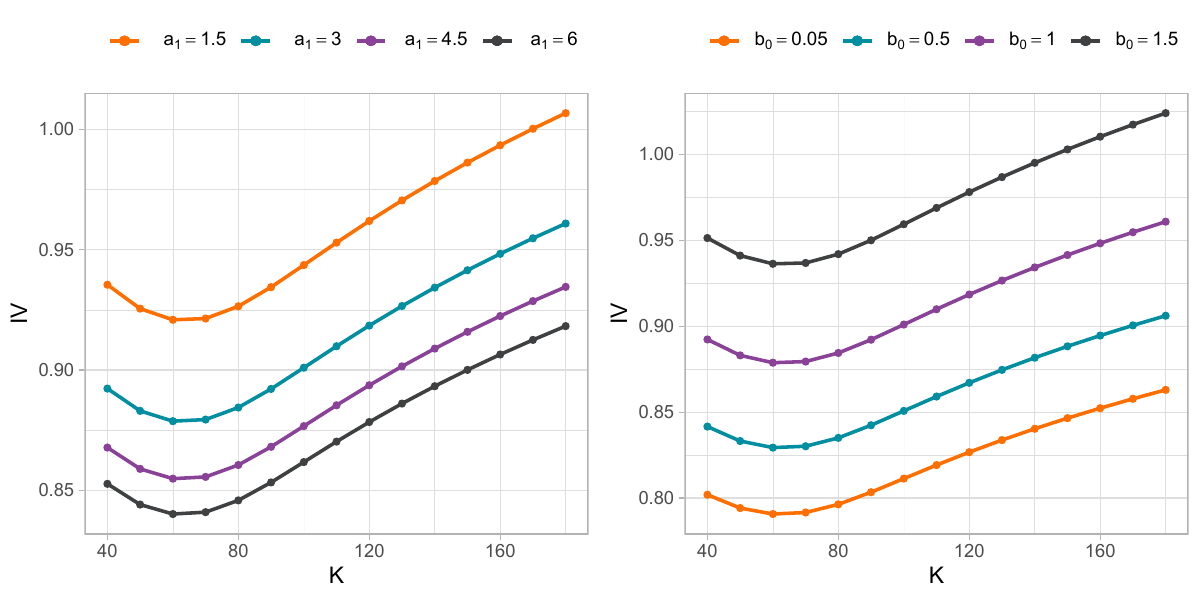}
		\caption{Hawkes model}
		\label{fig: carma_sens_hawkes}
	\end{subfigure}\hfill
	\begin{subfigure}{0.81\textwidth}
		\includegraphics[width=\textwidth]{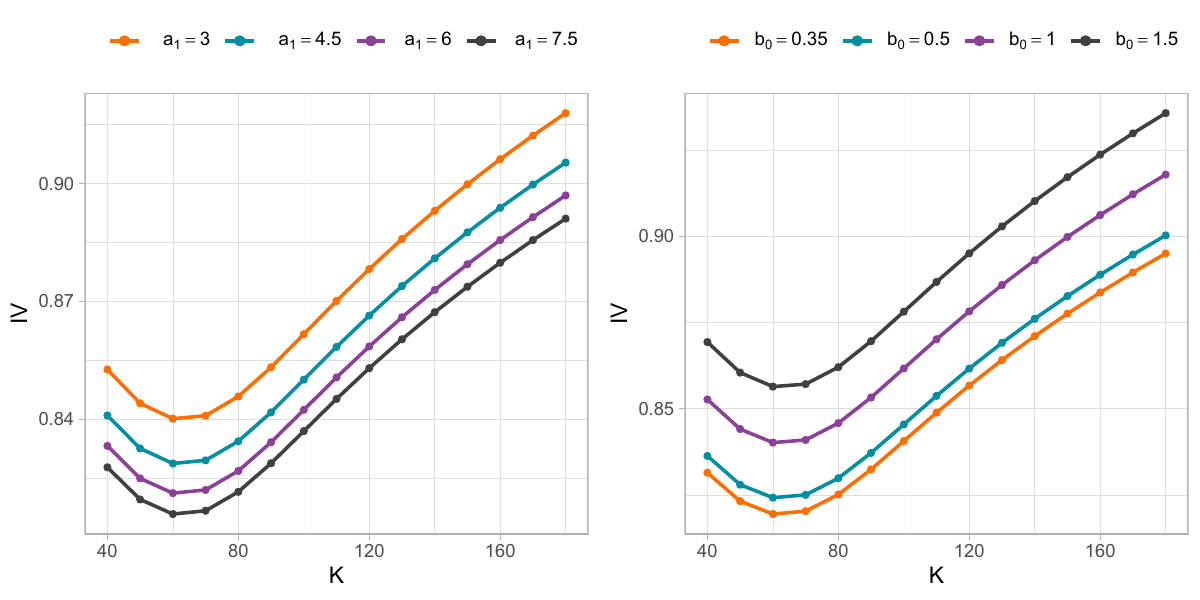}
		\caption{CARMA(2,1)-Hawkes model}
		\label{fig: carma_carma_21}
	\end{subfigure}\hfill
	\begin{subfigure}{0.81\textwidth}
		\includegraphics[width=\textwidth]{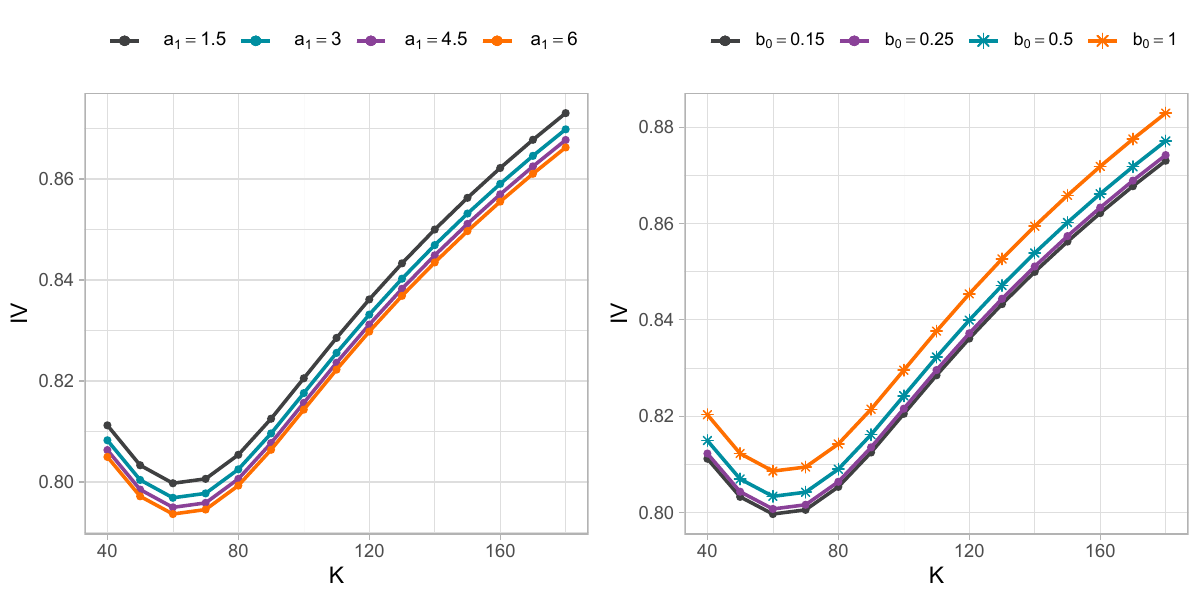}
		\caption{CARMA(3,1)-Hawkes model}
		\label{fig: carma_carma_32}
	\end{subfigure}
	\label{fig: carma_sens}
\end{figure}

%% file: empirical_analysis.tex

This section presents an empirical case study in which the calibration of Hawkes (used as a benchmark for comparison purposes), CARMA(2,1)-Hawkes, and CARMA(3,2)-Hawkes models to market data is conducted with the aim of elucidating the necessity for a generalisation of the Hawkes model. Furthermore, the additional objective is to determine which of the three models provides the optimal fit to the observed data in the context of an extreme scenario (non-regular market environment). To this end, a non-conventional company is selected for the calibration exercise: \textit{GameStop} (GME). 

The decision to select the aforementioned company was made on the grounds of \textit{Roaring Kitty}'s impact on \textit{GameStop} share price, which caused significant fluctuations in the stock price (see Figure~\ref{fig: gme_stock_price}) and thus created high levels of volatility. Such financial turbulence observed in GameStop stock price can be attributed to a post (without referring directly to the firm) on social media platforms on 13 May 2024, following a three-year absence from them, by \textit{Roaring Kitty}, an account linked to a social media finance influencer who is accounted for triggering the meme stock rally in 2021 (also referred as meme stock mania and GameStop short squeeze). The post and news about the \textit{Roaring Kitty}'s return created a sentiment of a bullish speculation around GameStop as in 2021, with speculators and investors (including also the reappearance of retail ones) increasing their activity and presence in the equity and derivatives markets. 

\begin{figure}[h]
	\centering
	\caption{The price performance of GameStop shares over a 12-month period. The gray area refers to the interval between the date on which the post was published and the date on which the dataset was collected.} 
	\includegraphics[width=0.7\textwidth]{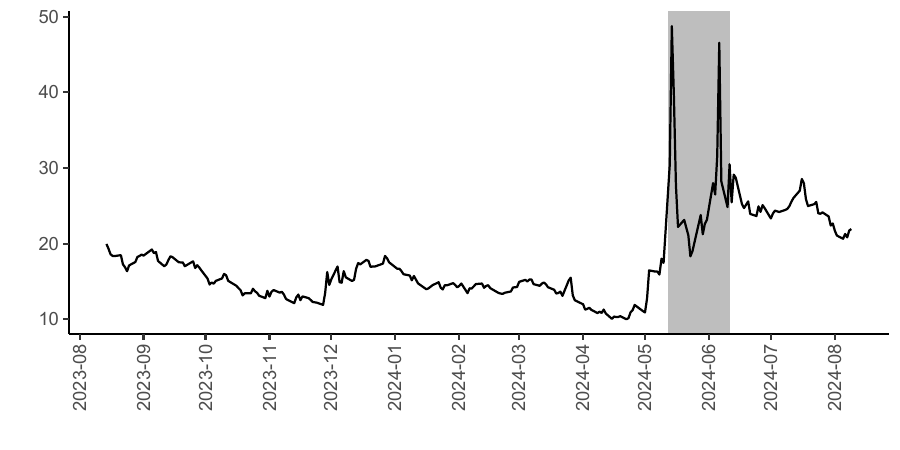} 
	\label{fig: gme_stock_price} 
\end{figure}

Having clarified this point, the data set used for the empirical study was extracted from the \textit{Chicago Board Options Exchange} (CBOE) on 11 June 2024 and consists of call options denominated in USD expiring on 16 August 2024 with a spot value ($S_0$) of $24.58$. The data set was filtered to exclude all options with open interest and volume values less than $10$, resulting in a final data set composed of 40 call options with a range of strike prices ($K$) between $12$ and $125$. In calibrating the models, we consider options that fall within the range of $-1.15$ to $0.75$ in terms of log-moneyness, defined as $k:= \ln\left(S_0/K\right)$, for a total of 31 call options with $K \in \left[12, 75\right]$. The number of in-the-money options is 10, while the number of out-of-the-money options is 21. The remaining call options, which are significantly deep out-of-the-money, are employed for an out-of-sample analysis, with a specific focus on strike prices $K \in \left\{80, 100, 125\right\}$. These strikes that are strongly above the spot value might be of particular interest to market practitioners due to the considerable high level of trading volume (respectively 148, 7748, and 1149), which justifies the analysis.




\subsection{Empirical results}\label{emp_results}

The model parameters are identified by minimising the relative root mean square error (RRMSE) between the market and the model-determined volatilities. Specifically, the minimisation problem is configured as follows. Let $\theta$ be a vector containing the baseline intensity parameter $\mu$, the autoregressive parameters $\left(a_1, \dots, a_p\right)$, and moving average parameters $\left(b_0, \dots, b_q\right)$ of a CARMA(p,q)-Hawkes model; i.e., $\theta = \left(\mu, b_0, \dots, b_q, a_1, \dots, a_p \right)$. For the sake of clarity, $\theta \in \Theta \subseteq \mathbb{R}^{p+q+2}$ where $\Theta$ denotes a compact subset of $\mathbb{R}^{p+q+2}$ in which the stationary condition is guaranteed and the kernel function is non-negative. Then, the set of parameters characterising a specific model (in this case the jump-diffusion model referenced in~\eqref{eq:STQ} with normally distributed jump sizes) is represented by $\mathbf{\Psi}:= \left(\theta, \mu_J, \sigma_J, \sigma\right)$ and the minimisation problem reads
\begin{equation}
	\mathbf{\Psi}^{\star} := \argmin_{\mathbf{\Psi} \in \mathbb{R}^{p+q+5}} \; f(\mathbf{\Psi}). 
\end{equation}
As the selected error measure is the RRMSE, it follows that the objective function $f(\mathbf{\Psi})$ writes
\begin{equation}
	f(\mathbf{\Psi}) := \sqrt{\frac{1}{n_K} \sum_{i = 1}^{n_K} \left(\frac{\text{IV}_{\text{model}} \left(K_i, \mathbf{\Psi}\right)-
			\text{IV}_{\text{market}}\left(K_i\right)}{\text{IV}_{\text{market}} \left(K_i\right)}\right)^2}, 
\end{equation}
where $n_K$ refers to the number of strike prices, $\text{IV}_{\text{market}} \left(K\right)$ is the market implied volatility with strike price $K$, and $\text{IV}_{\text{model}} \left(K, \mathbf{\Psi}\right)$ denotes the (Black-Scholes) implied volatility of a European call option with strike price $K$.


The calibration is implemented for the Hawkes, CARMA(2,1)-Hawkes, and CARMA(3,2)-Hawkes models;\footnote{In order to assess the suitability of choice $m = 450$
we evaluate prices of call options for the whole data in-the-sample analysis with an order $m$ relatively high ($m = 4000$). The absolute percentage error is less than $10^{-7}$ for all the three models. 
In light of the negligible nature of all errors, we decided to calibrate the models with $m = 450$ in order to accelerate the procedure. 
} the optimal parameters and the RRMSE $(\%)$ are presented in Table~\ref{tab: calibrated_params}.
\begin{table}[h]
	\centering
	\caption{Optimal parameters and RRMSE $(\%)$ determined through calibration for the Hawkes, CARMA(2,1)-Hawkes, and CARMA(3,2)-Hawkes models.}
	\label{tab: calibrated_params}
	\begin{tabular}{cccc}
		\toprule
		& Hawkes & CARMA(2,1)-Hawkes & CARMA(3,2)-Hawkes  \\ \midrule
		$\mu$      & 1.5325   & 1.4852                & 1.9283                    \\
		$b_0$      & 0.0007   & 0.5547                & 0.0346                  \\
		$b_1$      & ------   & 0.8823                & 0.7943                  \\
		$b_2$      & ------   & ------                & 0.9851                  \\
		$a_1$      & 11.6902  & 2.9850                & 5.1110                  \\
		$a_2$      & ------   & 1.4796                & 2.0242                   \\
		$a_3$      & ------   & ------                & 0.0348 \\
		$\mu_J$    & $0.5704$ & 0.8967                & 0.6719                   \\
		$\sigma_J$ & 0.8895   & 0.5702                & 0.6698                \\
		$\sigma$   & 0.7366   & 0.7317                & 0.6829                  \\\midrule
		RRMSE $(\%)$& 1.6647  & 0.7579                & 0.7492                \\ \bottomrule
	\end{tabular}
\end{table}

The findings of Table~\ref{tab: calibrated_params} show a clear reduction in the RRMSE for both the CARMA(2,1)-Hawkes (0.76\%) and CARMA(3,2)-Hawkes (0.75\%) models in comparison to the Hawkes model (1.66\%). In order to assess the suitability of a generalisation of the Hawkes model in terms of their ability to accommodate the observed data, two distinct analyses are conducted using the calibrated parameters in Table~\ref{tab: calibrated_params}. The former, which is called global sample analysis, involves
 the computation of the relative root mean square error for the entire data set, solely in-the-money (ITM) call options, solely out-of-the-money (OTM) call options, and a spectrum of log-moneyness values. The latter constitutes a proper out-of-sample analysis, entailing the computation of the RRMSE for the remaining data set (composed of 9 call options) and three strike prices selected for their high level of trading volume. The results are exhibited in Table~\ref{table: out_of_sample}.

\begin{table}[h]
	\centering
	\caption{Global sample and out-of-sample error (RRMSE $\%$) for the Hawkes, CARMA(2,1)-Hawkes, and CARMA(3,2)-Hawkes models. $k$ and $K$ are respectively log-moneyness and strike price.}
	\label{table: out_of_sample}
		\begin{tabular}{@{}ccccccccc@{}}
			\toprule
			\multirow{2}{*}{} & \multicolumn{4}{c}{Global sample}                             & \multicolumn{4}{c}{Out-of-sample (significantly deep OTM)} \\ \cmidrule(l){2-9} 
			& \begin{tabular}[c]{@{}c@{}}Entire\\ data set\end{tabular} & \begin{tabular}[c]{@{}c@{}}$k > 0$\end{tabular}  & \begin{tabular}[c]{@{}c@{}}$k < 0$ \end{tabular}  & \begin{tabular}[c]{@{}c@{}}$-0.35\leq k \leq 0.35$  \end{tabular} & \begin{tabular}[c]{@{}c@{}}Remaining\\ data set\end{tabular} & K = 80  & K = 100  & K = 125  \\ \midrule
			Number of calls    & 40          & 10       & 30    & 15                    & 9                & 1     & 1      & 1        \\ \midrule
			Hawkes             & 1.92        & 1.56     & 2.03  & 1.30                  & 2.62             & 1.41  & 2.70    & 3.46     \\
			CARMA(2,1)-Hawkes  & 0.93        & 0.80     & 0.97  & 0.81                  & 1.36             & 0.73  & 1.05   & 2.43   \\
			CARMA(3,2)-Hawkes  & 0.85        & 0.73     & 0.88  & 0.83                  & 1.12             & 0.94  & 0.72   & 1.97     \\ \bottomrule
		\end{tabular}
\end{table}

The results of Table~\ref{table: out_of_sample} indicate that the CARMA(3,2)-Hawkes model exhibits a lower relative root mean square error compared to the other two models in both analyses. However, there are two cases where the CARMA(2,1)-Hawkes model provides a lower RRMSE, namely in the case of  $-0.35\leq k \leq 0.35$ and when the strike price is equal to $80$. The noteworthy aspect is that both models, i.e., CARMA(2,1)-Hawkes and CARMA(3,2)-Hawkes, demonstrate a substantial enhancement in performance (in terms of RRMSE) with regard to that observed in the Hawkes model, suggesting that a generalisation of the Hawkes model is necessary in this specific instance.

Figure~\ref{fig_iv: comparison} illustrates a comparison of the market data and the calibrated implied volatilities for the Hawkes, CARMA(2,1)-Hawkes, and CARMA(3,2)-Hawkes models at varying strike prices, selected to represent the following cases: in-the-money, out-of-the-money, and deeply out-of-the-money. Figure~\ref{fig: cal_surf} shows the volatility surface for the Hawkes, CARMA(2,1)-Hawkes, and CARMA(3,2)-Hawkes models employing the optimal parameters identified through calibration.  
%
\begin{figure}[htbp!]
	\centering
	\caption{Comparison between market-implied volatilities and implied volatilities obtained with the calibrated Hawkes, CARMA(2,1)-Hawkes, and CARMA(3,2)-Hawkes models.} 
	\includegraphics[width=0.9\textwidth]{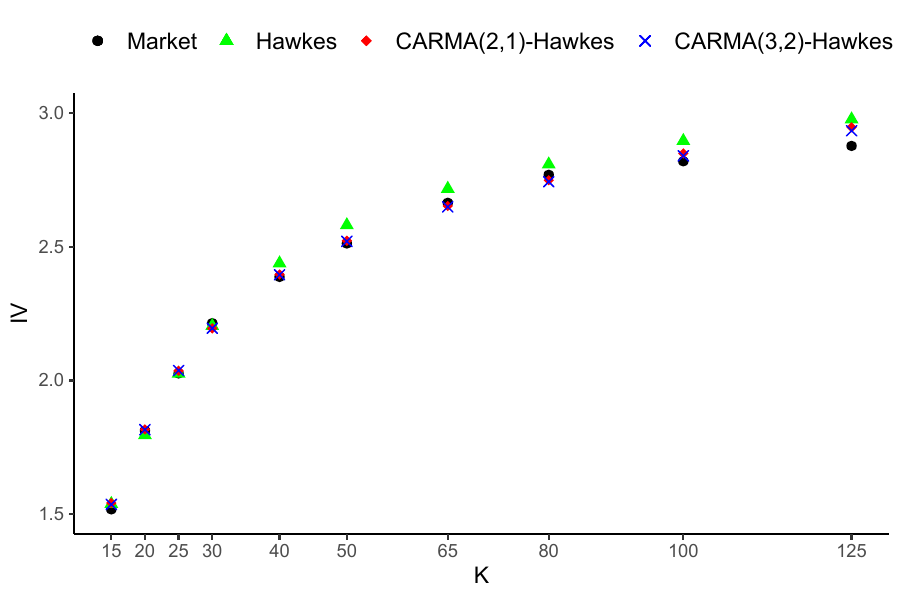} 
	\label{fig_iv: comparison} 
\end{figure}


\begin{figure}[htbp!]
	\centering
	\caption{Volatility surfaces using optimal parameters determined through calibration for the Hawkes, CARMA(2,1)-Hawkes, and CARMA(3,2)-Hawkes models.}
	\label{fig: cal_surf}
	\begin{subfigure}{0.6\textwidth}
		\includegraphics[width=\textwidth]{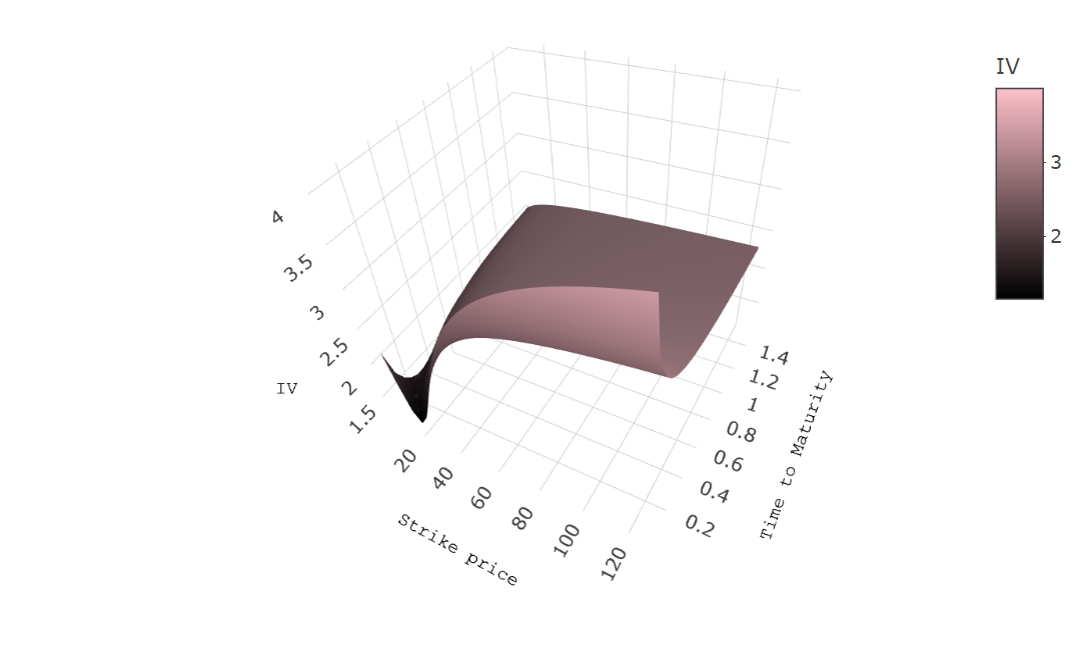}
		\caption{Hawkes model}
	\end{subfigure}\hfill
	\begin{subfigure}{0.6\textwidth}
		\includegraphics[width=\textwidth]{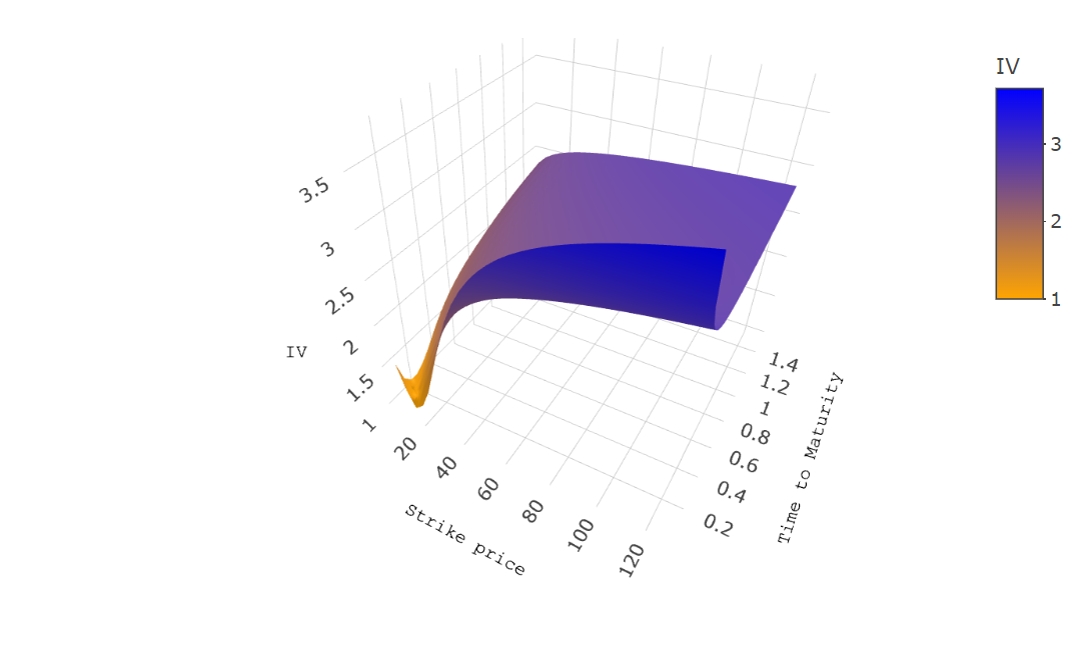}
		\caption{CARMA(2,1)-Hawkes model}
	\end{subfigure}\hfill
	\begin{subfigure}{0.6\textwidth}
		\includegraphics[width=\textwidth]{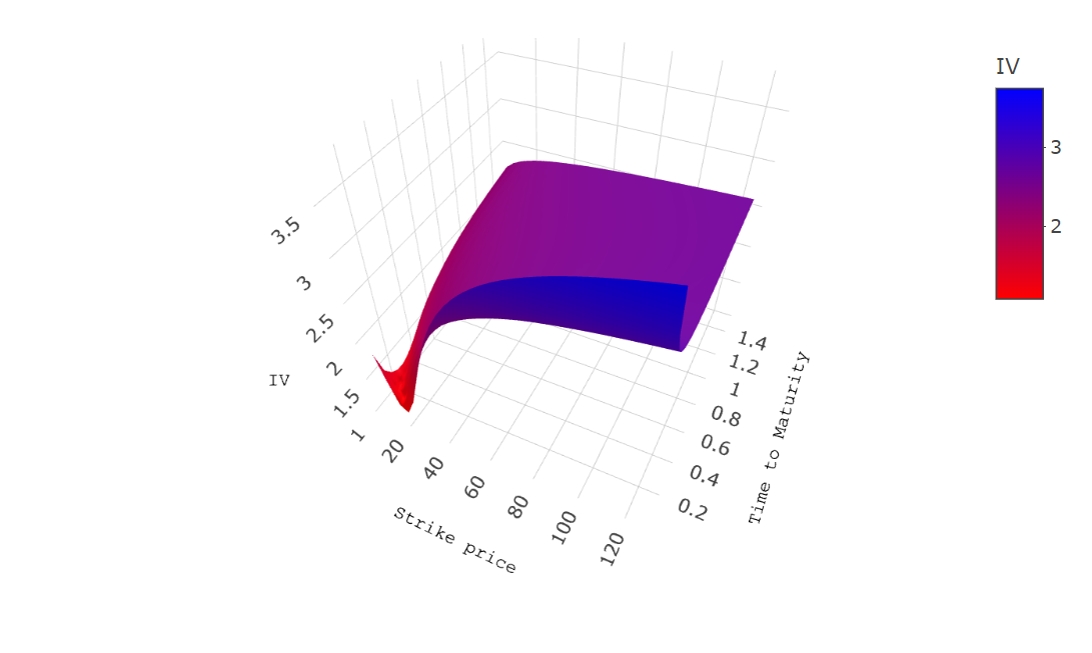}
		\caption{CARMA(3,2)-Hawkes model}
	\end{subfigure}
\end{figure}

%% file: Carma_review.tex
Here we discuss a simplified version of the model in \eqref{pureJumpwithnounitary}. We assume that, under the risk-neutral probability measure $\mathbb{Q}$, the jump size satisfies the following condition:
\begin{equation*}
\mathbb{E}^{\mathbb{Q}}\left[e^{\bar{J}^\mathbb{Q}}\right]=1.
\end{equation*}
Thus, the SDE in \eqref{sde_simplied_version} of the underlying asset price writes
\begin{equation}\label{sde_toymodel}
	\mbox{d} S_{t}=rS_{t^{-}}\mbox{d}t+S_{t^{-}}\mbox{d}\tilde{Y}^{\mathbb{Q}}_t.
\end{equation}
Given the initial condition $S_{t_0}$, the solution $S_{T}$ with $T\geq t_0$ reads
\begin{equation}
S_{T}=S_{t_0}\exp\left[r\left(T-t_0\right)+\sum_{k=N_{t_0}+1}^{N_{T}}J^\mathbb{Q}_k\right].
\label{SolutionToyMod1}
\end{equation}
Due to the Markov property of the CARMA(p,q)-Hawkes process, the solution in \eqref{SolutionToyMod1} can be rewritten in distribution as follows
\begin{equation*}
S_{T}\overset{d}{=}S_{t_0}\exp\left[r\left(T-t_0\right)+\sum_{k=1}^{N_{T}-N_{t_0}}J^\mathbb{Q}_k\right].
\end{equation*}
As discussed in \cite[Appendix D.1]{MERCURI20241}, the increments $N_{T}-N_{t_0}$ can be seen as the counting process in a CARMA(p,q)-Hawkes observed at time $T$, started at time $t_0$, where the state process has an initial value equal to $X_{t_0}$.

For the computation of European option prices, the Gauss-Laguerre quadrature discussed in Section~\ref{GaussLaguerre} is no longer required. Instead, it has a form similar to the Merton's jump-diffusion model~\cite{merton1976option}. We rewrite the option call price formula as follows:
\begin{eqnarray}
c\left(K, t_0, T\right)&=&e^{-r\left(T-t_0\right)}\mathbb{E}^{\mathbb{Q}}\left[\left(S_T-K\right)^{+}\left|\mathcal{F}_{t_0}\right.\right]\nonumber\\
&=&e^{-r\left(T-t_0\right)}\mathbb{E}^{\mathbb{Q}}\left[\mathbb{E}^{\mathbb{Q}}_{t_0}\left[\left(S_T-K\right)^{+}\left|N_{T}-N_{t_0}\right.\right]\left|\mathcal{F}_{t_0}\right.\right]\nonumber\\
&=&\sum_{n=0}^{+\infty} e^{-r\left(T-t_0\right)}\mathbb{E}^{\mathbb{Q}}_{t_0}\left[\left(S_T-K\right)^{+}\left|N_{T}-N_{t_0}=n\right.\right]\mathbb{Q}\left(N_{T}-N_{t_0}=n\left|\mathcal{F}_{t_0}\right.\right).
\label{PricingForToy}
\end{eqnarray}
Denoting with $Y^{\mathbb{Q}}_n :=\sum_{k=1}^{n}J^\mathbb{Q}_{{k}}$ for any $n\geq 1$ and $F_{Y^{\mathbb{Q}}_n}$ its cumulative distribution function, the conditional expectation in~\eqref{PricingForToy} becomes:
\begin{equation}\label{call_first}
e^{-r\left(T-t_0\right)}\mathbb{E}^{\mathbb{Q}}_{t_0}\left[\left(S_T-K\right)^{+}\left| N_{T}-N_{t_0}=n\right.\right]=S_{t_0}\int_{\ln\left(\frac{K}{S_{t_0}}\right)-r\left(T-t_0\right)}^{+\infty}e^{u}\mbox{d}F_{Y^{\mathbb{Q}}_n}\left(u\right)-Ke^{-r \left(T-t_0\right)}\int_{\ln\left(\frac{K}{S_{t_0}}\right)-r\left(T-t_0\right)}^{+\infty}\mbox{d}F_{Y^{\mathbb{Q}}_n}\left(u\right).
\end{equation}
Observing that $\mathbb{E}^{\mathbb{Q}}\left[Y^{\mathbb{Q}}_n \right]=1$, we introduce a new probability measure $\bar{F}_{Y^{\mathbb{Q}}_n}\left(y\right)$ defined as
\begin{equation}\label{cdf}
\bar{F}_{Y^{\mathbb{Q}}_n}\left(y\right):= \int_{-\infty}^{y} e^u \mbox{d}F_{Y^{\mathbb{Q}}_n}\left(u\right).
\end{equation}
Using~\eqref{cdf} in \eqref{PricingForToy}, we get
\begin{equation}
e^{-r\left(T-t_0\right)}\mathbb{E}^{\mathbb{Q}}_{t_0}\left[\left(S_T-K\right)^{+}\left|N_{T}-N_{t_0}=n\right.\right]=S_0\left[1-\bar{F}_{Y^{\mathbb{Q}}_n}\left(\bar{d}\right)\right]-Ke^{-r\left(T-t_0\right)}\left[1-F_{Y^{\mathbb{Q}}_n}\left(\bar{d}\right)\right],
\label{cond_Exp}
\end{equation}
where $\bar{d}:=\ln\left(\frac{K}{S_{t_0}}\right)-r\left(T-t_0\right)$. Substituting \eqref{cond_Exp} into \eqref{PricingForToy}, we get:
\begin{eqnarray}\label{pricing_formula_appendix}
c\left(K, t_0, T\right)&=& e^{-r\left(T-t_0\right)}\left(S_{t_0}e^{r\left(T-t_0\right)}-K\right)^{+}\mathbb{Q}\left(N_{T}-N_{t_0}=0\left|\mathcal{F}_{t_0}\right.\right)\nonumber\\
&+&\sum_{n=1}^{+\infty}\left\{S_0\left[1-\bar{F}_{Y^{\mathbb{Q}}_n}\left(\bar{d}\right)\right]-Ke^{-r\left(T-t_0\right)}\left[1-F_{Y^{\mathbb{Q}}_n}\left(\bar{d}\right)\right]\right\}\mathbb{Q}\left(N_{T}-N_{t_0}=n\left|\mathcal{F}_{t_0}\right.\right).
\end{eqnarray}

In order to apply the pricing formula in~\eqref{pricing_formula_appendix}, we need to compute the following quantities: $\mathbb{Q}\left(N_{T}-N_{t_0}=n\left|\mathcal{F}_{t_0}\right.\right)$, $\bar{F}_{Y^{\mathbb{Q}}_n }\left(\bar{d}\right)$, and $F_{Y^{\mathbb{Q}}_n }\left(\bar{d}\right)$. The last two quantities depend on the specification of the jump size distribution under the risk-neutral probability measure $\mathbb{Q}$; some possible choices are discussed at the end of the appendix. The probability $\mathbb{Q}\left(N_{T}-N_{t_0}=n\left|\mathcal{F}_{t_0}\right.\right)$ can be calculated using two approaches: the conditional probability generating function of a point process or the Fourier transform for a discrete random variable.

As the conditional probability generating function (refer to  \cite{kocherlakota2017bivariate} for further details) of $N_{T}-N_{t_0}$ can be defined as
\begin{equation}
\mathcal{G}_{T,t_0}\left(z\right):=\mathbb{E}^{\mathbb{Q}}\left(z^{N_{T}-N_{t_0}}\left|\mathcal{F}_{t_0}\right.\right),
\label{pgfCP}
\end{equation} 
the probability of events $\left\{N_{T}-N_{t_0}=n\right\}_{n \geq 0}$ writes
\begin{equation}
\mathbb{Q}\left(N_{T}-N_{t_0}=n\left|\mathcal{F}_{t_0}\right.\right)=\frac{1}{n!}\frac{\partial^n \mathcal{G}_{T,t_0}\left(z\right)}{\left(\partial z\right)^{n}}.
\end{equation}
For the second approach, the probability $\mathbb{Q}\left(N_{T}-N_{t_0}=n\left|\mathcal{F}_{t_0}\right.\right)$ can be computed in the following way: 
\begin{equation}
\mathbb{Q}\left(N_{T}-N_{t_0}=n\left|\mathcal{F}_{t_0}\right.\right) = \frac{1}{2\pi}\int_{0}^{2\pi}e^{-i n \kappa} \phi_{N_T - N_{t_0}}\left(\kappa\right)    \mbox{d}\kappa
\end{equation}
where the conditional characteristic function $\phi_{N_T - N_{t_0}}\left(\kappa\right)$  is defined as 
\begin{equation}\label{chf_simple_model}
\phi_{N_T - N_{t_0}}\left(\kappa\right) := \mathbb{E}^\mathbb{Q}\left[e^{i \kappa \left(N_T - N_{t_0}\right)}\left|\mathcal{F}_{t_0}\right.\right]. 
\end{equation}
It is worth to notice that the conditional probability generating function~\eqref{pgfCP} can be obtained from the conditional characteristic function~\eqref{chf_simple_model}; i.e., $\mathcal{G}_{T,t_0}\left(z\right) = \phi_{N_T - N_{t_0}}\left(-i \ln z\right)$.

The following lemma provides a log-affine form formula for joint conditional characteristic function of the vector $\left(X_{T},  N_{T}\right)$, which is required for the computation of \eqref{chf_simple_model}.
\begin{lemma}\label{thmJointCharNX}
Let $\phi_{X_T, N_T}\left(u, \kappa, t_0\right)$ be the CARMA(p,q)-Hawkes conditional joint characteristic function given the information at time $t_0<T$ defined as:
\begin{equation}\label{chf_lemma}
\phi_{X_T, N_T}\left(u,\kappa,t_0\right):= \mathbb{E}^{\mathbb{Q}}\left[e^{i u^{\top}X_T+i \kappa N_{T}}\left|\mathcal{F}_{t_0}\right.\right],
\end{equation}
with $u\in \mathbb{R}^p$ and $\kappa\in \mathbb{R}$.
The function $\phi_{X_T, N_T}\left(u,\kappa,t_0\right)$ satisfies the following log-affine form:
\begin{equation}
\phi_{X_T, N_T}\left(u,\kappa,t_0\right) := \exp\left[u_{0, T}\left(t_0\right)+u_{T}\left(t_0\right)^{\top}X_t+\kappa_T\left(t_0\right) N_{t_0}\right],
\label{LogLin}
\end{equation}
where the time coefficients $u_{0, T}\left(\cdot\right)$, $u_{T}\left(\cdot\right)$ and $\kappa_T\left(\cdot\right)$ are the solution of the following system of ordinary differential equations:
\begin{equation}
\left\{
\begin{array}{l}
\frac{\partial\kappa_T\left(t\right)}{\partial t}=0\\
\frac{\partial u_{0, T}\left(t\right)}{\partial t} = \mu\left(1-e^{u_{T}\left(t\right)^{\top}\mathbf{e}+   \kappa_T\left(t\right)}\right)\\
\mathbb{J}_{u_{T}\left(t\right)}^{\top} = \left(1-e^{u_{T}\left(t\right)^{\top}\mathbf{e}+\kappa_T\left(t\right)}\right)\mathbf{b}^{\top}-u_{T}\left(t\right)^{\top}\mathbf{A}
\end{array}
\label{eq:FinalRes1}
\right.,
\end{equation}
with the final conditions:
\begin{equation}
\left\{
\begin{array}{l}
u_{0,T} \left(T\right) = 0\\
u_{T} \left(T\right) = iu\\
\kappa_{T} \left(T\right) = i\kappa
\end{array}
\right. .
\label{finalcond1}
\end{equation}
\end{lemma}
\begin{proof}
We consider the stochastic process $\left\{\phi_{X_T, N_T}\left(u, \kappa, t\right)\right\}_{t \in \left[t_0, T\right]}$ that is a complex martingale\footnote{
Let $\left\{\phi_{X_T, N_T}\left(u, \kappa, t\right)\right\}_{t \in \left[t_0, T\right]}$ be a stochastic process defined as
\begin{equation*}
  \phi_{X_T, N_T}\left(u, \kappa, t\right):= \mathbb{E}^\mathbb{Q}\left[\exp\left(i u^{\top}X_T+i \kappa N_{T}\right)\left|\mathcal{F}_{t}\right.\right], \quad \forall t \in \left[t_0, T\right].
\end{equation*}
In order to show that the process is a martingale, we use its definition and the tower property of the conditional expected value. Let $s,t \in \left[t_0, T\right]$ with $s < t$. Then
\begin{eqnarray*}
\mathbb{E}^\mathbb{Q}\left[\phi_{X_T, N_T}\left(u,\kappa, t\right)\left|\mathcal{F}_s\right.\right] &=& \mathbb{E}^\mathbb{Q}\left[\mathbb{E}^\mathbb{Q}\left[\phi_{X_T, N_T}\left(u,\kappa, T\right)\left|\mathcal{F}_t\right.\right]\left|\mathcal{F}_s\right.\right]\\
&=& \mathbb{E}^\mathbb{Q}\left[\phi_{X_T, N_T}\left(u,\kappa, T\right)\left|\mathcal{F}_s\right.\right]\\
&=& \phi_{X_T, N_T}\left(u,\kappa, s\right).
\end{eqnarray*}
} 
with the final value at time $T$ given by:
\begin{equation*}
\phi_{X_T, N_T}\left(u,\kappa, T\right)=e^{i u^{\top} X_T+ i \kappa N_T}.
\end{equation*}
We assume that $\phi_{X_T, N_T}\left(u,\kappa,t\right)$ has a log-affine structure in the components $\left(X_t,N_t\right)$. Specifically $\forall t\in\left[t_0,T\right]$:
\begin{equation*}
\phi_{X_T, N_T}\left(u,\kappa,t\right) := \exp\left[u_{0, T}\left(t\right)+u_{T}\left(t\right)^{\top}X_t+\kappa_T\left(t\right) N_t\right], 
\end{equation*}
where the time-coefficients $u_{0, T}\left(\cdot\right),u_{T}\left(\cdot\right)$ and $\kappa_T\left(\cdot\right)$ satisfy the final conditions
\begin{equation*}
\left\{
\begin{array}{l}
u_{0,T} \left(T\right) = 0\\
u_{T} \left(T\right) = u\\
\kappa_{T} \left(T\right) = \kappa
\end{array}
\right. .
\end{equation*}
The martingale property of the process $\left\{\phi_{X_T, N_T}\left(u,\kappa,t\right)\right\}_{t\in\left[t_0, T\right]}$ implies:
\begin{equation}
\frac{\partial \phi_{X_T, N_T}\left(u,\kappa,t\right)}{\partial t}+\mathcal{A}\phi_{X_T, N_T}\left(u,\kappa,t\right)=0,
\label{eq:PDE1}
\end{equation}
with the final conditions in \eqref{finalcond1}. $\mathcal{A}\phi_{X_T, N_T}\left(u,\kappa,t\right)$ is the infinitesimal generator of the CARMA(p,q)-Hawkes \citep[see][for more details]{MERCURI20241}, that is:
\begin{eqnarray*}
\mathcal{A}\phi_{X_T, N_T}\left(u,\kappa,t\right) &=& \left(\mu+\mathbf{b}^{\top}X_t\right)\phi_{X_T, N_T}\left(u,\kappa,t\right)\left(e^{u_{T}\left(t\right)^{\top}\mathbf{e}+\kappa_T\left(t\right)}-1\right)+\phi_{X_T, N_T}\left(u, \kappa, t\right)u_{T}\left(t\right)^{\top}\mathbf{A}X_t\nonumber\\
&=& \phi_{X_T, N_T}\left(u,\kappa,t\right) \left\{\mu\left(e^{u_{T}\left(t\right)^{\top}\mathbf{e}+\kappa_T\left(t\right)}-1\right)+\left[\left(e^{u_{T}\left(t\right)^{\top}\mathbf{e}+\kappa_T\left(t\right)}-1\right)\mathbf{b}^{\top}+u_{T}\left(t\right)^{\top}\mathbf{A}\right]X_t\right\}.
\end{eqnarray*}
The partial derivative of $\phi_{X_T, N_T}\left(u,\kappa,t\right)$ with respect to time $t$ reads
\begin{equation*}
\frac{\partial \phi_{X_T, N_T}\left(u,\kappa,t\right)}{\partial t} = \phi_{X_T, N_T}\left(u,\kappa, t\right) \left[\frac{\partial u_{0, T}\left(t\right)}{\partial t}+ \mathbb{J}_{u_T\left(t\right)}^{\top}X_t+\frac{\partial\kappa_T\left(t\right)}{\partial t} N_t\right].
\end{equation*}
The equation in \eqref{eq:PDE1} becomes:
\begin{equation*}
\left[\frac{\partial u_{0, T}\left(t\right)}{\partial t}-\mu\left(1-e^{u_{T}\left(t\right)^{\top}\mathbf{e}+\kappa_T\left(t\right)}\right)\right]+ \left[\mathbb{J}_{u_T\left(t\right)}^{\top}-\left(1-e^{u_{T}\left(t\right)^{\top}\mathbf{e}+\kappa_T\left(t\right)}\right)\mathbf{b}^{\top}+u_{T}\left(t\right)^{\top}\mathbf{A}\right]X_t+\frac{\partial\kappa_T\left(t\right)}{\partial t}N_t=0,
\end{equation*}
and this result leads to the system in \eqref{eq:FinalRes1}, concluding the proof.
\end{proof}

\begin{remark}
Note that the joint characteristic function~\eqref{chf_simple_model} can be rewritten using \eqref{chf_lemma}. Specifically,
\begin{equation}\label{phi1_remark}
	\phi_{N_T - N_{t_0}}\left(\kappa\right) = \phi_{N_T - N_{t_0}}\left(\mathbf{0}, \kappa, t_0\right) e^{-i\kappa N_{t_0}}.
\end{equation}
Using Lemma~\ref{thmJointCharNX}, we have
\begin{equation}\label{phi2_remark}
	\phi_{N_T - N_{t_0}}\left(\mathbf{0}, \kappa, t_0\right) = \exp\left(u_{0, T}\left(t_0\right) +  u_{T}\left(t_0\right)^{\top} X_{t_0} + i \kappa N_{t_0}\right).
\end{equation}
Substituting~\eqref{phi2_remark} in \eqref{phi1_remark}, we finally
\begin{equation}
   \phi_{N_T - N_{t_0}}\left(\kappa\right) = \exp\left(u_{0, T}\left(t_0\right) +  u_{T}\left(t_0\right)^{\top} X_{t_0} \right).
\end{equation}
\end{remark}

\noindent In the case of normally distributed jump sizes, i.e., $\bar{J}^{\mathbb{Q}}\sim  \mathcal{N} \left(-\frac{1}{2}\sigma_J^2,\sigma_J^2\right)$, and recalling that $Y^{\mathbb{Q}}_n :=\sum_{k=1}^{n}J^\mathbb{Q}_{{k}}$ for any $n\geq 1$, the distribution of $Y^{\mathbb{Q}}_n$ is normally distributed with mean $-\frac{1}{2}\sigma_J^2  n$ and variance $\sigma_J^2  n$. \newline
The cumulative distribution functions $F_{Y^{\mathbb{Q}}_n }\left(\bar{d}\right)$ in \eqref{call_first} and $\bar{F}_{Y^{\mathbb{Q}}_n }\left(\bar{d}\right)$ defined in \eqref{cdf} can be written in terms of a standard normal cdf $\Phi\left(\cdot\right)$; that is:
\begin{equation}\label{cdf_normal}
	F_{Y^{\mathbb{Q}}_n }\left(\bar{d}\right) = \Phi \left(\frac{\bar{d} +\frac12 \sigma_J^2 n}{\sigma_J \sqrt{n}}\right) \quad \text{and} \quad \bar{F}_{Y^{\mathbb{Q}}_n }\left(\bar{d}\right) = \Phi \left(\frac{\bar{d} - \frac12 \sigma_J^2 n}{\sigma_J \sqrt{n}}\right).
\end{equation}
\begin{remark}
 Generalizing the model in \eqref{sde_toymodel} by adding a diffusive component on the price dynamics, the SDE writes  
\begin{equation*}
	\mbox{d} S_{t}=rS_{t^{-}}\mbox{d}t+ \sigma S_{t^{-}}\mbox{d}W_t + S_{t^{-}}\mbox{d}\tilde{Y}^{\mathbb{Q}}_t,
\end{equation*}
and its solution $S_T$, given the initial condition $S_{t_0}$, with $t_0 \geq T$ is
\begin{equation*}
	S_{T} = S_{t_0}\exp\left\{\left(r - \frac{1}{2} \sigma \right)\left(T-t_0\right) + \sigma W_T  + \sum_{k=1}^{N_{T}-N_{t_0}}J^\mathbb{Q}_{k}\right\}.
\end{equation*}
The pricing formula in \eqref{pricing_formula_appendix} can still be used where the quantities $\left[1- \bar{F}_{Y^{\mathbb{Q}}_n }\left(\bar{d}\right)\right]$ and $\left[1 - F_{Y^{\mathbb{Q}}_n }\left(\bar{d}\right)\right]$ are substituted respectively by $\text{N}\left(d_1\right)$ and $\text{N}\left(d_2\right)$ in the Black-Scholes formula with volatility $\sigma\left(n\right):= \sqrt{\sigma^2+n\sigma_J^2\frac{1}{T-t_0}}$.
\end{remark}

Now we present the case of shifted gamma distributed jump sizes. 
The condition $\mathbb{E}^\mathbb{Q}\left[\bar{J}^\mathbb{Q}\right] = 1$ implies that the jump size $\bar{J}^\mathbb{Q}$ is defined as $\bar{J}^\mathbb{Q}:= \Gamma_{\alpha, \beta} + \alpha \ln\left(1 - \frac{1}{\beta}\right)$ where the random variable $\Gamma_{\alpha, \beta}$ is gamma distributed with shape $\alpha>0$ and rate $\beta>1$. The quantity $Y^{\mathbb{Q}}_n$ becomes
\begin{equation*}	
	Y^{\mathbb{Q}}_n  = \Gamma_{\alpha n , \beta} + \alpha n \ln\left(1 - \frac{1}{\beta}\right),
\end{equation*}
which is the convolution of $n$ i.i.d. gamma random variables with the same distribution $\Gamma\left(\alpha, \beta\right)$. \newline 
The cumulative distribution function $F_{Y^{\mathbb{Q}}_n }\left(\bar{d}\right)$ in~\eqref{pricing_formula_appendix} can be written as
\begin{equation*}
F_{Y^{\mathbb{Q}}_n }\left(\bar{d}\right) =   \mathbb{Q}\left(\Gamma_{\alpha n , \beta}  \leq \bar{d} -\alpha n \ln\left(1 - \frac{1}{\beta}\right)\right).
\end{equation*}
In order to compute $\bar{F}_{Y^{\mathbb{Q}}_n }\left(\bar{d}\right)$, we consider the following expression:
\begin{eqnarray*}	
	\int_{\alpha n \ln\left(1 - \frac{1}{\beta}\right)}^{\infty} e^{cy} \mbox{d}\bar{F}_{Y^{\mathbb{Q}}_n}\left(y\right) =  \int_{\alpha n \ln\left(1 - \frac{1}{\beta}\right) }^{\infty} e^{cy} e^{y}\mbox{d} F_{Y^{\mathbb{Q}}_n}\left(y\right) 
\end{eqnarray*}
where the rhs is obtained using~\eqref{cdf}. We observe that the rhs is the moment generating function of $Y^{\mathbb{Q}}_n$ at $c+1$, that is:
\begin{eqnarray*}	
\int_{\alpha n \ln\left(1 - \frac{1}{\beta}\right)}^{\infty} e^{cy} e^{y}\mbox{d} F_{Y^{\mathbb{Q}}_n}\left(y\right) = \mathbb{E}^\mathbb{Q} \left[e^{\left(c+1\right)Y^{\mathbb{Q}}_n}  \right].
\end{eqnarray*}
The moment generating function of $Y^{\mathbb{Q}}_n$ for shifted gamma jump sizes is 
\begin{eqnarray*}
	\mathbb{E}^\mathbb{Q} \left[e^{\left(c+1\right)Y^{\mathbb{Q}}_n}  \right] 
	=\exp\left\{-\alpha n \ln\left(1 - \frac{c}{\beta - 1}\right) + \alpha c  n \ln\left(1 - \frac{1}{\beta}\right) \right\},
\end{eqnarray*}
that coincides with that of a shifted gamma with shape $\alpha n$ and new rate $\beta - 1$. 
Thus, $\bar{F}_{Y^{\mathbb{Q}}_n }\left(\bar{d}\right)$ becomes 
\begin{eqnarray*}
	\bar{F}_{Y^{\mathbb{Q}}_n} (\bar{d}) = 
     \mathbb{Q}\left( \Gamma_{\alpha n, \beta - 1}  \leq \bar{d} -\alpha n \ln\left(1 - \frac{1}{\beta}\right)\right).
\end{eqnarray*}